\documentclass[1p]{elsarticle}

\usepackage{amssymb}
\usepackage{amsfonts}
\usepackage{amstext}
\usepackage{graphicx}
\usepackage{amscd}
\usepackage{amsmath}
\usepackage{mathrsfs}
\usepackage{stmaryrd}

\newcommand{\ihbar}{\imath \hbar}

\newcommand{\Ran}{\mathrm{Ran}}
\newcommand{\coker}{\mathrm{coker}}
\newcommand{\coRan}{\mathrm{coRan}}
\newcommand{\Ted}{\underset{\rightarrow}{\mathbb{T}e}}
\newcommand{\Teg}{\underset{\leftarrow}{\mathbb{T}e}}
\newcommand{\Ped}{\underset{\rightarrow}{\mathbb{P}e}}
\newcommand{\HS}{{\mathcal H_{\mathcal S}}}
\newcommand{\HA}{{\mathcal H_{\mathcal A}}}
\newcommand{\tr}{\mathrm{tr}}
\newcommand{\llangle}{\langle \hspace{-0.2em} \langle}
\newcommand{\rrangle}{\rangle \hspace{-0.2em} \rangle}
\newcommand{\rdagger}{{\rotatebox{180}{\scriptsize $\dagger$}}}
\newcommand{\PSI}{\mathbf{\Psi}}
\newcommand{\Aut}{\mathrm{Aut}}
\newcommand{\id}{\mathrm{id}}
\newcommand{\RHO}{\wp}
\newcommand{\dist}{\mathrm{dist}}
\newcommand{\Obj}{\mathrm{Obj}}
\newcommand{\Morph}{\mathrm{Morph}}
\newcommand{\Funct}{\mathrm{Funct}}
\newcommand{\EndFunct}{\mathrm{EndFunct}}
\newcommand{\pinv}{\star}
\newcommand{\transp}{\text{\textsc{t}}}
\newcommand{\Verylongrightarrow}{= \hspace{-0.2em} = \hspace{-0.2em} \Longrightarrow}
\newcommand{\Verylongleftarrow}{\Longleftarrow \hspace{-0.2em} = \hspace{-0.2em} =}
\newcommand{\Verylongeq}{= \hspace{-0.2em} = \hspace{-0.2em} = \hspace{-0.2em} =}
\newcommand{\upcircarrow}{\overset{\uparrow}{\bullet}}
\newcommand{\Sp}{\mathrm{Sp}}

\newtheorem{defi}{Definition}
\newtheorem{prop}{Property}
\newtheorem{propo}{Proposition}
\newenvironment{proof}{\noindent \textit{Proof:}}{\hfill $\Box$ \\}
\newtheorem{cor}{Corollary}
\newtheorem{theo}{Theorem}

\journal{Journal of Geometry and Physics}

\begin{document}
\begin{frontmatter}

\title{Purification of Lindblad dynamics, geometry of mixed states and geometric phases}

\author[uti]{David Viennot}
\address[uti]{Institut UTINAM (CNRS UMR 6213, Universit\'e de Bourgogne-Franche-Comt\'e, Observatoire de Besan\c con), 41bis Avenue de l'Observatoire, BP 1615, 25010 Besan\c con cedex, France.}

\begin{abstract}
We propose a nonlinear Schr\"odinger equation in a Hilbert space enlarged with an ancilla such that the partial trace of its solution obeys to the Lindblad equation of an open quantum system. The dynamics involved by this nonlinear Schr\"odinger equation constitutes then a purification of the Lindblad dynamics. We study the (non adiabatic) geometric phases involved by this purification and show that our theory unifies several definitions of geometric phases for open systems which have been previously proposed. We study the geometry involved by this purification and show that it is a complicated geometric structure related to an higher gauge theory, i.e. a categorical bibundle with a connective structure.
\end{abstract}

\end{frontmatter}

\tableofcontents

\section{Introduction}
Quantum information\cite{Heiss} and open quantum systems\cite{Breuer} are subjects of particular interest in the modern physics, dealing with decoherence processes, quantum computation and communication, entanglement processes, distillation protocols, Schmidt decomposition, Markovian and non-Markovian effects, etc. A particular interesting subject in quantum information is the process of purification\cite{Bengtsson} which consists for a mixed state $\rho$ of the Hilbert space $\HS$ to find a pure state $\Psi \in \HS \otimes \HA$ in an enlarged Hilbert space such that $\rho = \tr_{\HA} |\Psi \rrangle \llangle \Psi|$. The auxiliary Hilbert space $\HA$ can be viewed as describing an effective environment. In this paper, we want to study the purification process with respect to the dynamics of mixed states. When the dynamics is conservative, i.e. when it is described by a Liouville-von Neumann equation $\ihbar \dot \rho = [H_{\mathcal S},\rho]$ (no relaxation effect occurs), the dynamics of the purification and the related mathematical structures have been extensively studied, see for example ref.\cite{Andersson1, Dittmann}. We study in this paper the case where the environment of the quantum system induces relaxation effects, with a dynamics obeying to a Lindblad equation\cite{Breuer}. We will show that the purified state obeys in the enlarged Hilbert space to a nonlinear Schr\"odinger equation. The emergence of a nonlinearity is not a new phenomenon in the relation between dynamics of mixed and pure states. In ref. \cite{Gisin} it is shown that the pure dynamics closest to the Lindblad dynamics (in the sense that this pure dynamics is viewed as the dynamics of some tangent vectors on the density matrix manifold) is nonlinear; and in ref. \cite{Bechmann}, a purification protocol needing a nonlinear operation has been proposed (a purification protocol is a set of operations and measurements transforming a mixed state $\rho$ to a pure state $\psi$ of the same Hilbert space without the trace operation -- in fact the state $\psi$ depends only on the protocol and is independent from the initial mixed state $\rho$ --, it is a question different from the purification process discussed in the present paper). Moreover the Liouville equation for a piecewise deterministic process is associated for its deterministic part with a nonlinear Schr\"odinger equation but which does not take into account the jump part (see ref.\cite{Breuer} chapter 6.1) in contrast with our equation for the purified dynamics.\\
Geometrization of physical theories is a great active area in theoretical physics. In nonrelativistic quantum dynamics, it is in particular related to the theory of geometric phases (so-called Berry phases) \cite{Berry}. As shown by Simon\cite{Simon}, the dynamics of pure states and the Berry phases take place in a geometric structure which is a principal fibre bundle endowed with a connection. Some generalisations of geometric phases have been proposed for mixed states: Uhlmann\cite{Uhlmann0,Uhlmann1, Uhlmann2} has proposed a concept of geometric phases based on the theory of transition probabilities for pair of mixed states, the involved geometric structures have been analysed in ref. \cite{Dittmann, Dittmann2, Jencova, Gimento}; Sj\"oqvist etal\cite{Sjoqvist, Tong} have proposed a concept of geometric phase based on an interferometric theory, the involved geometric structures have been analysed in ref. \cite{Andersson1, Andersson}; and we have proposed a concept of geometric phase in the adiabatic limit based on the generalisation of the geometric structure studied by Simon from vector bundles to $C^*$-modules \cite{Viennot1, Viennot2}. By using the equation of the purified dynamics, we will build a general theory of geometric phases for open quantum systems which unifies these previous approaches. The dynamics of the density matrices present different geometric phases appearing at different levels. This is due to a more complicated gauge structure called higher gauge theory in the literature \cite{Baez1, Baez2, Baez3, Wockel, Nikolaus}. We will show that this structure is a generalization in the category theory of the principal bundle structure, complicated by the stratified structure of the density matrix manifold\cite{Bengtsson}.\\
The approach followed in the present paper could be usefull for some problems of quantum control and quantum information concerning open systems. In ref. \cite{Viennot3} we have shown how use the fields associated with the geometry of categorical bundles to analyse the control of a quantum system hampered by entanglement with another one. By the purification of the Lindblad equation, the decoherence and the relaxtion effects on the mixed state of the open quantum system, appears in purified picture as entanglement between the quantum state and the ancilla  (described by the auxiliary Hilbert space). The fields associated with the categorical bundles presented in this paper, can then be used to analyse the control of an open quantum system (in the purified picture) with a similar method as the one followed in ref. \cite{Viennot3}. Moreover some approaches of quantum computation based on the geometrization of the quantum dynamics and on the geometric phases have been proposed \cite{Zanardi, Lucarelli, Sjoqvist2, Xu}, usually called holonomic quantum computation (HQC). These approaches are based on the geometric properties of fiber bundles modelizing the dynamics of closed quantum systems. The present work, with the construction of the categorical bundles describing open quantum systems in the purified picture, is a first step to a generalisation of the HQC taking into account the decoherence and the relaxation effects occuring for the open systems.\\
This paper is organized as follows. Section II is devoted to the geometry associated with the purification process by recalling the stratified structure of the density matrix manifold and by introducing some representations of the purified states with the associated inner products. This section introduces some mathematical tools needed to the understanding of the following. Section III shows that the purified state satisfies a nonlinear Schr\"odinger equation if the associated mixed state satisfies a Lindblad equation. Section IV presents a general theory of geometric phases for open quantum systems; there relations with the different propositions of geometric phases are analysed. 
Section V studies the geometric structure involved by the general geometric phase theory and the purification process, firstly from the point of view of the ordinary differential geometry, secondly from the point of view of the category theory. It concludes by the introduction of the connective structure and the physical interpretations of the different fields involved by the connection.\\

\noindent {\it \textbf{A note about the notations used here:}\\
We adopt the Einstein's convention: a bottom-top repetition of an index induces a summation.}\\
{\it  $\mathcal B(\mathcal H)$ denotes the set of the bounded linear operators of the Hilbert space $\mathcal H$. For an operator $A \in \mathcal B(\mathcal H)$, $\Ran A$, $\ker A$ and $\Sp(A)$ denote its range, kernel and spectrum. $\forall A,B \in \mathcal B(\mathcal H)$, we denote the commutator and the anticommutator by $[A,B]=AB-BA$ and $\{A,B\}=AB+BA$. $\Aut \mathfrak g$, with $\mathfrak g$ a vector space or an algebra, denotes the set of the automorphisms of $\mathfrak g$.}\\
{\it The symbol ``$\simeq$'' between two spaces (resp. manifolds) denotes that they are isomorphic (resp. homeomorphic). The symbol ``$\approx$'' between two manifolds denotes that they are locally homeomorphic. The symbol ``$A \hookrightarrow B$'' denotes an inclusion of $A$ into $B$. $G \rtimes H$ denotes a semi-direct product between two groups $G$ and $H$; $\mathfrak g \niplus \mathfrak h$ denotes a semi-direct sum between two algebras $\mathfrak g$ and $\mathfrak h$.}\\
{\it $\Pr_i : V_1 \times V_2 \times ... \times V_n \to V_i$, with $V_j$ some sets, denotes the canonical projection $\Pr_i(x_1,x_2,...,x_n) = x_i$. Let $M$ be a manifold, $T_xM$ denotes its tangent space at $x$ ($TM$ denotes its tangent bundle) and $\Omega^n(M,\mathfrak g)$ denotes its space of $\mathfrak g$-valued differential $n$-forms. Let $f: M \to N$ be a diffeomorphism between two manifolds, $f_* : TM \to TN$ denotes its tangent map (its push-forward) and $f^* : \Omega^* N \to \Omega^* M$ denotes its cotangent map (its pull-back).} Let $E \xrightarrow{\pi} M$ be a fibre bundle ($E$ and $M$ are manifolds and $\pi$ is a surjective map); $\Gamma(M,E)$ denotes the set of its local sections.\\
{\it For a category $\mathscr C$, $\Obj \mathscr C$ denotes its collection of objects and $\Morph \mathscr C$ denotes its collection of arrows (morphisms). $\forall o \in \Obj \mathscr C$, $\id_o$ denotes the trivial arrow from and to $o$ (the identity map of $o$). $\forall a \in \Morph \mathscr C$, $s(a)$ denotes the source of $a$, and $t(a)$ denotes the target of $a$. $\Funct(\mathscr C,\mathscr C')$ denotes the set of functors from $\mathscr C$ to $\mathscr C'$ ($\EndFunct(\mathscr C) \equiv \Funct(\mathscr C,\mathscr C)$).}\\
{\it $\forall A(t) \in \mathcal B(\mathcal H)$, $\Teg^{-\int_{t_0}^t A(t')dt'} = U_A(t,t_0)$ denotes the time-ordered exponential (the Dyson series) i.e. the solution of the equation: $ \frac{dU_A(t,t_0)}{dt} = -A(t) U_A(t,t_0)$ (with $U_A(t_0,t_0) = \id_{\mathcal H}$). $\Ted^{-\int_{t_0}^t A(t')dt'} = V_A(t_0,t)$ denotes the time-anti-ordered exponential, i.e. the solution of the equation: $ \frac{dV_A(t_0,t)}{dt} = - V_A(t_0,t) A(t)$ (with $V_A(t_0,t_0) = \id_{\mathcal H}$).}

\section{Purification process}
\subsection{The purification bundle and the stratification}
Let $\HS$ be the Hilbert space of the studied system $\mathcal S$ (we consider that $\HS$ is finite dimensional, $\HS \simeq \mathbb C^n$, since the applications of the present work concern essentially quantum information theory where the qubit state spaces are finite dimensional). Mixed states of $\mathcal S$ are represented by density operators; the space of density operators being
\begin{equation}
\mathcal D_0(\HS) = \{ \rho \in \mathcal B(\HS), \rho^\dagger = \rho, \rho \geq 0, \tr \rho = 1 \}
\end{equation}
where $\mathcal B(\HS) \simeq \mathfrak M_{n\times n}(\mathbb C)$ denotes the space of linear operators of $\HS$. The mixed state $\omega_\rho : \mathcal B(\HS) \to \mathbb C$ associated with $\rho \in \mathcal D_0(\HS)$ is defined by
\begin{equation}
\forall A \in \mathcal B(\HS), \quad \omega_\rho(A) = \tr(\rho A)
\end{equation}
$\mathcal B(\HS)$ being a $C^*$-algebra and $\omega_\rho$ being a state of this $C^*$-algebra (see ref.\cite{Bratteli}). In some cases, we can also consider non-normalized density operators:
\begin{equation}
\mathcal D(\HS) = \{ \rho \in \mathcal B(\HS), \rho^\dagger = \rho, \rho \geq 0 \}
\end{equation}
\begin{equation}
\forall A \in \mathcal B(\HS), \quad \omega_\rho(A) = \frac{\tr(\rho A)}{\tr \rho}
\end{equation}
A purification of a density operator $\rho \in \mathcal D(\HS)$ is a state $\Psi_{\rho} \in \HS \otimes \HA$ such that
\begin{equation}
\rho = \tr_{\HA} |\Psi_\rho \rrangle \llangle \Psi_\rho |
\end{equation}
where $\HA$ is an arbitrary auxiliary Hilbert space called the \textit{ancilla} and $\llangle .|. \rrangle$ denotes the inner product of $\HS \otimes \HA$ induced by the tensor product ($\tr_\HA$ denotes the partial trace over $\HA$). Because of the Schmidt theorem\cite{Bengtsson} one needs $\dim \HA \geq \dim \HS$. In order to avoid any unnecessary complication, we choose $\dim \HA = \dim \HS$ throughout this paper.\\
An interesting choice of ancilla consists to the algebraic dual of $\HS$, $\HA = \HS^* = \{\langle \psi|, \psi \in \HS\}$ (the space of continuous linear functionals of $\HS$). In that case $\HS \otimes \HA = \HS \otimes \HS^* = \mathcal B(\HS)$. This choice is called \textit{standard purification}. $\mathcal B(\HS)$ is endowed with the Hilbert-Schmidt inner product:
\begin{equation}
\forall W,Z \in \mathcal B(\HS), \quad \langle Z|W \rangle_{HS} = \tr_\HS(Z^\dagger W) = \overline Z_{i\alpha} W^{i \alpha}
\end{equation}
where $W = W^{i\alpha} |\zeta_i \rangle \langle \zeta_\alpha|$, $(\zeta_i)_{i=1,...,n}$ being an orthonormal basis of $\HS$. We have then
\begin{equation}
\begin{array}{rcl} \pi_{HS} : \mathcal B(\HS) & \to & \mathcal D(\HS) \\ W & \mapsto & WW^\dagger \end{array}
\end{equation}
The restriction of $\pi_{HS}$ on $S_{HS} \mathcal B(\HS) = \{W \in \mathcal B(\HS), \|W\|_{HS}^2 = \tr_\HS(W^\dagger W) = 1 \}$ has its values in $\mathcal D_0(\HS)$. $\pi_{HS}$ is a surjective map.
\begin{equation}
\forall W\in S_{HS}\mathcal B(\HS), \forall U \in U(\HS), \qquad \tilde W = WU \Rightarrow \pi_{HS}(\tilde W) = WUU^\dagger W^\dagger = \pi_{HS}(W)
\end{equation}
$U(\HS) \simeq U(n)$ denotes the group of unitary operators of $\HS$. The action of $U(\HS)$ on $S_{HS}\mathcal B(\HS)$ is not transitive except for the restriction on the faithfull operators ($\det W \not= 0$).\\
$ U(\HS)$ acts also on the left of $S_{HS}\mathcal B(\HS)$. This action induces the adjoint action of the group on $\mathcal D_0(\HS)$:
\begin{equation}
\forall W\in S_{HS}\mathcal B(\HS), \forall U \in U(\HS), \qquad \tilde W = UW \Rightarrow  \pi_{HS}(\tilde W) = U \pi_{HS}(W) U^{-1}
\end{equation}
$\forall \rho \in \mathcal D_0(\HS)$, the orbit of the density operator $U(\HS)\rho = \{U\rho U^{-1}, U\in U(\HS)\}$ is constituted by all operators which are isospectral to $\rho$. $\Sigma(\HS) = \mathcal D_0(\HS) / U(\HS) \simeq \Sigma(n)$ where
\begin{equation}
\Sigma(n) = \{(p_1,...,p_n) \in [0,1]^n, p_i \leq p_{i+1}, \sum_{i=1}^n p_i = 1 \}
\end{equation}
is the $(n-1)$-simplex of the possible eigenvalues ($\dim_{\mathbb R} \Sigma(n) = n-1$). $\Sigma(\HS)$ is the set of diagonal density operators with sorted eigenvalues. Let $\pi_D: \mathcal D_0(\HS) \to \Sigma(\HS)$ be the canonical projection associated with the quotient space $\mathcal D_0(\HS) / U(\HS)$.
\begin{equation}
\forall \sigma \in \Sigma(\HS), \quad \pi_D^{-1}(\sigma) = \{U\sigma U^{-1}, U \in U(\HS)\} \simeq U(\HS)/U(\HS)_\sigma
\end{equation}
where $U(\HS)_\sigma = \{U \in U(\HS), U\sigma U^{-1} = \sigma \}$ is the isotropy subgroup (the stabilizer) of $\sigma$. The bundle $\mathcal D_0(\HS) \xrightarrow{\pi_D} \Sigma(\HS)$ is not locally trivial, it is a stratified space\cite{Bengtsson, Bredon}. A stratum $\Sigma^i(\HS)$ of $\Sigma(\HS)$ is characterized by the degeneracy profil of the eigenvalues. If $\sigma \in \Sigma^i(\HS)$ has $k_1$ eigenvalues with degeneracy equal to $q_1$, $k_2$ eigenvalues with degeneracy equal to $q_2$, etc, then its stabilizer is $U(\HS)_\sigma \simeq U(q_1)^{k_1} \times ... \times U(q_l)^{k_l}$ (the products are the group direct products) and then $\pi_D^{-1}(\sigma) \simeq U(n)/(U(q_1)^{k_1} \times ... \times U(q_l)^{k_l})$ ($\sum_i k_iq_i = n$). The stratum $\Sigma^\circ(\HS)$ with no eigenvalue degeneracy is associated with the normalizer $U(\HS)_\sigma \simeq T^n$ (the $n$-torus) and $\forall \sigma \in \Sigma^\circ(\HS)$, $\pi_D^{-1}(\sigma) \simeq U(n)/T^n = Fl(n,\mathbb C)$ (the flag manifold); in the other cases, the fibers of the strata in $D_0(\HS)$ are homeomorphic to partial flag manifolds as Grassmanian manifolds $Gr_p(\mathbb C^n) = U(n)/(U(n-p)\times U(p))$, projective manifolds $\mathbb CP^{n-1} = U(n)/(U(n-1) \times U(1))$, etc. We distinguish between the regular strata such that $\det \rho \not= 0$ (no zero eigenvalue) and the singular strata such that $\det \rho =0$. We remark the presence of two particular strata associated with vertices of $\Sigma(n)$, the singular stratum of the pure states $\Sigma^0(\HS)$ ($p_1=...=p_{n-1}=0$ and $p_n = 1$, $\pi_D^{-1}(\sigma) \simeq U(n)/(U(n-1)\times U(1)) = \mathbb CP^{n-1}$), and the regular stratum of the microcanonical state $\Sigma^\infty(\HS)$ ($p_1=...=p_n = \frac{1}{n}$, $\pi_D^{-1}(\sigma) \simeq  U(n)/U(n)= \{1\}$). Fig. \ref{simplex} illustrates the stratification for the smallest dimensional cases.
\begin{figure}
\begin{center}
\includegraphics[width=14cm]{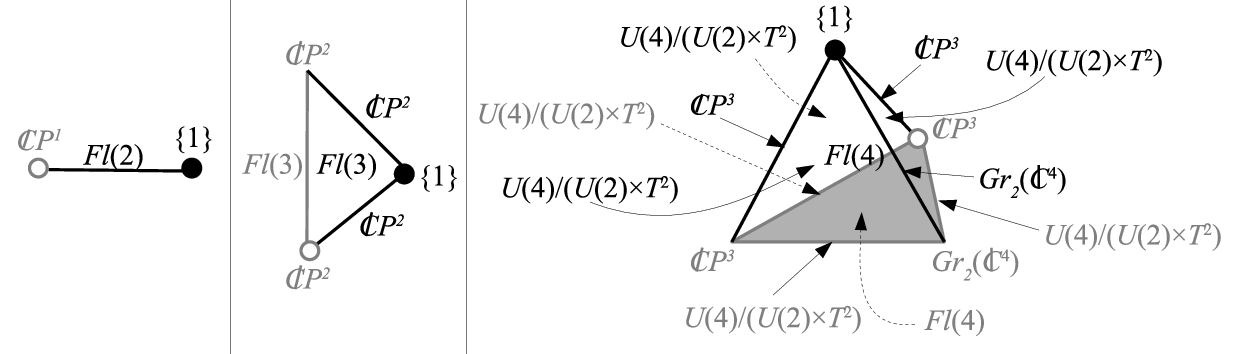}
\caption{\label{simplex} The simplex $\Sigma(n)$ for $n=2$ (left), $n=3$ (middle) and $n=4$ (right). The strata are the open sub-simplices of each simplex (including $\Sigma^\circ(n)$ the interior of $\Sigma(n)$ -- the stratum with no degenerated eigenvalues --). For each stratum $\Sigma^i$ the model manifold of $\pi^{-1}_D(\sigma)$ is indicated. Singular strata are written in grey whereas regular strata are written in black ($\bullet$ denotes the stratum of the microcanonical state, and $\circ$ denotes the stratum of the pure states). The description of the simplexes can be found table \ref{simpl_desc}. The details about the stratification can found in ref.\cite{Bengtsson} where $\Sigma(n)$ is called Weyl chamber by analogy with the group representation theory.}
\end{center}
\end{figure}

\begin{table}
\begin{center}
\caption{\label{simpl_desc} Descriptions of the stratification of the simplex $\Sigma(n)$.}
\begin{tabular}{|c|c|c|c|}
\hline
$n$ & strata $\Sigma^i(n)$ & sub-simplices & fiber $\pi_D^{-1}(\sigma)$ \\
\hline \hline
2 & $0<p_1<p_2$ & edge & $Fl(2,\mathbb C)$  \\
\cline{2-4}
 & $p_1=p_2=1/2$ & vertex & $\{1\}$  \\
\cline{2-4}
 & $p_1=0<p_2=1$ & vertex & $\mathbb CP^1$ \\
\hline 
3 & $0<p_1<p_2<p_3$ & facet & $Fl(3,\mathbb C)$ \\
\cline{2-4}
& $0<p_1 = p_2 < p_3$ & edges & $\mathbb CP^2$ \\
& $0<p_1 < p_2 = p_3$ & & \\
\cline{2-4}
& $p_1=0 < p_2 < p_3$ & edge & $Fl(3,\mathbb C)$ \\
\cline{2-4}
& $p_1=0 <p_2=p_3=1/2$ & vertices & $\mathbb CP^2$ \\
& $p_1=p_2=0<p_3=1$ & & \\
\cline{2-4}
& $p_1=p_2=p_3=1/3$ & vertex & $\{1\}$ \\
\hline
4 & $0<p_1<p_2<p_3<p_4$ & cell & $Fl(4,\mathbb C)$ \\
\cline{2-4}
& $0<p_1=p_2<p_3<p_4$ & facets & $U(4)/(U(2)\times T^2)$ \\
& $0<p_1<p_2=p_3<p_4$ & & \\
& $0<p_1<p_2<p_3=p_4$ & & \\
\cline{2-4}
& $0=p_1<p_2<p_3<p_4$ & facet & $Fl(4,\mathbb C)$ \\
\cline{2-4}
& $0<p_1=p_2<p_3=p_4$ & edge & $Gr_2(\mathbb C^4)$ \\
\cline{2-4}
& $0<p_1=p_2=p_3<p_4$ & edges & $\mathbb CP^3$ \\
& $0<p_1<p_2=p_3=p_4$ & & \\
\cline{2-4}
& $p_1=0<p_2<p_3=p_4$ & edges & $U(4)/(U(2)\times T^2)$ \\
& $p_1=0<p_2=p_3<p_4$ & & \\
& $p_1=p_2=0<p_3<p_4$ & & \\
\cline{2-4}
& $p_1=p_2=p_3=0<p_4=1$ & vertices & $\mathbb CP^3$ \\
& $p_1=0<p_2=p_3=p_4=1/3$ & & \\
\cline{2-4}
& $p_1=p_2=0<p_3=p_4=1/2$ & vertex & $Gr_2(\mathbb C^4)$ \\
\cline{2-4}
& $p_1=p_2=p_3=p_4=1/4$ & vertex & $\{1\}$ \\
\hline
\end{tabular}
\end{center}
\end{table}

Let $\HA$ be an arbitrary ancilla, and let $(\zeta_i)_{i=1,...,n}$ and $(\xi_\alpha)_{\alpha=1,...,n}$ be orthonormal basis of $\HS$ and $\HA$. Let $\PSI : \mathcal B(\HS) \xrightarrow{\simeq} \HS \otimes \HA$ be the (non canonical) isomorphism defined by
\begin{equation}
\forall W \in \mathcal B(\HS), \quad \PSI(W) = W^{i \alpha} \zeta_i \otimes \xi_\alpha
\end{equation}
with $W = W^{i \alpha} |\zeta_i \rangle \langle \zeta_\alpha|$. The restriction of $\PSI$ on $S_{HS} \mathcal B(\HS)$ has its values in $S(\HS \otimes \HA) = \{ \Psi \in \HS \otimes \HA, \|\Psi\|_{\HS \otimes \HA} = 1 \} \simeq S^{2n^2-1}$ (the $(2n^2-1)$-dimensional sphere). Note that $\llangle \PSI(Z) | \PSI(W) \rrangle = \langle Z|W \rangle_{HS}$.\\
The following commutative diagramm summarizes the bundle of purification:
$$ \begin{CD}
S_{HS} \mathcal B(\HS) @>{\simeq}>{\Psi}> S(\HS \otimes \HA) @>{\simeq}>> S^{2n^2-1} \\
@V{\pi_{HS}}VV  @V{\pi_{\mathcal A}}VV  @VVV \\
\mathcal D_0(\HS)  @= \mathcal D_0(\HS)  @>{\simeq}>> \mathcal D_0(\mathbb C^n)  @= \bigsqcup_i \mathcal D^i_0(\mathbb C^n) @<{\hookleftarrow}<< U(n)/U(n)_\sigma \\
@V{\pi_D}VV   @V{\pi_D}VV  @VVV @VVV @VVV \\
\Sigma(\HS)  @=  \Sigma(\HS)  @>{\simeq}>> \Sigma(n) @= \bigsqcup_i \Sigma^i(n) @<{\hookleftarrow}<< \{\sigma\}
\end{CD}$$ 
where $\hookleftarrow$ denotes inclusion maps and $\simeq$ denotes homeomorphisms, the vertical arrows are projections.\\
The ancilla plays the role of an effective environment with which the system is entangled. The right action of $U(\HS)$ on $S_{HS}\mathcal B(\HS)$ is associated with unitary transformations of $\HA$.
\begin{equation}
\forall U \in U(\HS), \forall W \in S_{HS}\mathcal B(\HS), \quad \PSI(WU) = 1_\HS \otimes U^\transp \PSI(W) \Rightarrow \pi_\Psi(\PSI(WU)) = \rho_W
\end{equation}
where $U^\transp = \langle \zeta_\beta|U\zeta^\alpha \rangle |\xi_\alpha \rangle \langle \xi^\beta| \in U(\HA)$ ($\transp$ denotes the transposition) and  $\rho_W = \pi_{\mathcal A}(\PSI(W)) = \pi_{HS}(W) = WW^\dagger$. Such a transformation has no consequence for the observator which has only information concerning the system (information concerning the environment is lost by the partial trace $\tr_\HA$). It is then associated with an inobservable reconfiguration of the ancilla (the effective environment). In contrast, the left action of $U(\HS)$ on $S_{HS}\mathcal B(\HS)$ is associated with unitary transformations of $\HS$ modifying the mixed state:
\begin{equation}
\forall U \in U(\HS), \forall W \in S_{HS}\mathcal B(\HS), \quad \PSI(UW) = U \otimes 1_\HA \Psi(W) \Rightarrow \pi_\Psi(\PSI(UW)) = U \rho_W U^{-1}
\end{equation}
In quantum information, it can be interesting to consider also SLOCC transformations (Stochastic Local Operations and Classical Communication). By virtue of the principle of the quantum open systems, information concerning the ancilla (the environment) is lost by taking the partial trace (information stored in the environment is not directly accessible), the physicist cannot performs SLOCC transformations on $\HA$. The group of SLOCC transformations on $\HS$ is $SL(\HS) \simeq SL(n,\mathbb C)$ (the group of invertible operators of $\HS$ with determinant equal to 1). Since $U(\HS) \supset SU(\HS) \subset SL(\HS)$ ($SU(\HS)$ is the group of unitary operators of $\HS$ with determinant equal to 1), the group of unitary and SLOCC transformations are $SL(\HS) \times U(1) = GL(\HS)$ ($U(1)$ is the group of phase changes and $GL(\HS)$ is the group of invertible operators of $\HS$).
\begin{equation}
\forall U \in GL(\HS), \forall W \in \mathcal B(\HS), \quad \PSI(UW) = U \otimes 1_\HA \PSI(W) \Rightarrow \pi_\Psi(\PSI(UW)) = U \rho_W U^\dagger
\end{equation}
by noting that $S_{HS} \mathcal B(\HS)$ and $\mathcal D_0(\HS)$ are not stable by the right action of $SL(\HS)$ (we must then consider $\mathcal B(\HS)$ and $\mathcal D(\HS)$). We extend $\pi_D$ on the whole of $\mathcal D(\HS)$ by $\forall U \in SL(\HS)$, $\forall \rho \in \mathcal D_0(\HS)$, $\pi_D(U\rho U^\dagger) = \pi_D(\rho)$.

\subsection{The $C^*$-module structure}
The purification space $\HS \otimes \HA$ can also be viewed as a left Hilbert $\mathcal B(\HS)$-$C^*$-module\cite{Landsman}, the left action of the $C^*$-algebra $\mathcal B(\HS)$ being defined by
\begin{equation}
\forall A \in \mathcal B(\HS), \forall \Psi \in \HS \otimes \HA, \quad A \Psi = A \otimes 1_\HA \Psi = \Psi^{i\alpha} (A\zeta_i) \otimes \xi_\alpha
\end{equation}
The inner product of the $C^*$-module is defined by
\begin{equation}
\forall \Psi,\Phi \in \HS \otimes \HA, \quad \langle \Phi|\Psi \rangle_* = \tr_\HA |\Psi \rrangle \llangle \Phi| \in \mathcal B(\HS)
\end{equation}
It satisfies the usual properties:
\begin{equation}
\forall A,B \in \mathcal B(\HS), \forall \Psi,\Phi,\Xi \in \HS \otimes \HA, \quad \langle \Phi|A\Psi+B\Xi \rangle_* = A \langle \Phi|\Psi\rangle_* + B \langle \Phi|\Xi \rangle_*
\end{equation}
\begin{equation}
\forall \Psi,\Phi \in \HS \otimes \HA, \quad \langle \Phi|\Psi \rangle^\dagger_* = \langle \Psi|\Phi \rangle_*
\end{equation}
\begin{equation}
\forall \Psi \in \HS \otimes \HA, \quad \langle \Psi|\Psi \rangle_* \geq 0 \text{ and }  \langle \Psi|\Psi \rangle_* = 0 \iff \Psi = 0
\end{equation} 
The density operator associated with a state of $\HS \otimes \HA$ is its square norm with respect to the $C^*$-module structure:
\begin{equation}
\forall \Psi \in \HS \otimes \HA, \quad \|\Psi\|^2_* = \langle \Psi|\Psi \rangle_* = \rho_\Psi \in \mathcal D(\HS)
\end{equation}

$\mathcal B(\HS)$ as the standard purification space is related to the $C^*$-module structure by
\begin{equation}
\forall A \in \mathcal B(\HS), \forall W \in \mathcal B(\HS), \quad \PSI(AW) = A \PSI(W)
\end{equation}
\begin{equation}
\forall W,Z \in \mathcal B(\HS), \quad \langle \PSI(W)|\PSI(Z) \rangle_* = ZW^\dagger
\end{equation}
The right action of $\mathcal B(\HS)$ on itself induces a particular class of operators of the $C^*$-module acting only on $\HA$:
\begin{equation}
\forall A \in \mathcal B(\HS), \forall W \in \mathcal B(\HS), \quad \PSI(WA) = 1_\HS\otimes A^\transp \PSI(W)
\end{equation}

\subsection{$C^*$-adjointness}
In this part we want to define the adjoint, with respect to the inner product of the $C^*$-module, of the right action of $\mathcal B(\HS)$ on $\HS \otimes \HA$.
\begin{defi}
Let $\Gamma \in \mathcal B(\HS)$. We define $\Gamma^\rdagger: \HS \otimes \HA \to \mathcal B(\HS)$ by
\begin{eqnarray}
\forall \Psi \in \HS \otimes \HA, \quad \Gamma^\rdagger(\Psi) & = & {\Gamma^i}_j \Psi^{j \alpha} (W^\pinv_\Psi)_{\beta i} |\xi_\alpha \rangle \langle \xi^\beta| \\
& = & {(W^\pinv_\Psi \Gamma W_\Psi)_\beta}^\alpha |\xi_\alpha \rangle \langle \xi^\beta| \\
& = & {{(W^\pinv_\Psi \Gamma W_\Psi)^\transp}^\alpha}_\beta |\xi_\alpha \rangle \langle \xi^\beta|
\end{eqnarray}
where $W_\Psi \in \mathcal B(\HS)$ is such that $\PSI(W_\Psi)=\Psi$, ${\Gamma^i}_j = \langle \zeta^i|\Gamma \zeta_j \rangle$, $\Psi^{j \alpha} = W^{j \alpha}_\Psi = \langle \zeta^j | W_\Psi \zeta^\alpha \rangle = \llangle \zeta^j \otimes \xi^\alpha|\Psi \rrangle$, and where $W^\pinv_\Psi \in \mathcal B(\HS)$ is the pseudo-inverse of $W_\Psi$:
\begin{equation}
W^\pinv_\Psi W_\Psi = 1 - P_{\ker W_\Psi} \text{ and } W_\Psi W^\pinv_\Psi = P_{\Ran W_\Psi}
\end{equation}
where $P_{\ker W_\Psi}$ and $P_{\Ran W_\Psi}$ are respectively the projections on the kernel and on the range of $W_\Psi$ in $\HS$.
\end{defi}
\begin{prop}
Let $\Gamma^\ddagger(\Psi) = \Gamma^\rdagger(\Psi)^\dagger$ (the adjoint $\dagger$ is in the sense of the inner product of $\HA$). We have
\begin{eqnarray}
\forall \Phi \in \HS\otimes \HA, & \quad & \langle \Phi|1_\HS \otimes \Gamma^\rdagger(\Psi) \Psi \rangle_* = \langle 1_\HS \otimes \Gamma^\ddagger(\Psi) \Phi|\Psi \rangle_* \\
& \iff & \langle \Phi|P_{\Ran W_\Psi} \Gamma \otimes 1_\HA \Psi \rangle_* = \langle 1_\HS \otimes \Gamma^\ddagger(\Psi) \Phi|\Psi \rangle_*
\end{eqnarray}
\end{prop}
For this reason, we call $\Gamma^\ddagger$ the $C^*$-adjoint of $\Gamma$.\\
\begin{cor}
\begin{equation}
\forall \Psi \in \HS \otimes \HA, \quad 1_\HS \otimes \Gamma^\rdagger(\Psi)\Psi = P_{\Ran W_\Psi} \Gamma \otimes 1_\HA \Psi
\end{equation}
\end{cor}
\begin{proof}
\begin{eqnarray}
1_\HS \otimes  \Gamma^\rdagger(\Psi)\Psi & = & {\Gamma^\rdagger(\Psi)^\alpha}_\beta \Psi^{i \beta} \zeta_i \otimes \xi_\alpha \\
& = & {\Gamma^k}_j \Psi^{j\alpha} (W_\Psi^\pinv)_{\beta k} \Psi^{i \beta} \zeta_i \otimes \xi_\alpha \\
& = & {\Gamma^k}_j \Psi^{j\alpha} {(W_\Psi W_\Psi^\pinv)^i}_k \zeta_i \otimes \xi_\alpha \\
& = & {(P_{\Ran W_\Psi})^i}_k {\Gamma^k}_j \Psi^{j \alpha} \zeta_i \otimes \xi_\alpha \\
& = & P_{\Ran W_\Psi} \Gamma \otimes 1_\HS \Psi
\end{eqnarray}

\begin{eqnarray}
\langle \Phi|1_\HS \otimes \Gamma^\rdagger(\Psi)\Psi \rangle_* & = & \tr_\HA |1_\HS \otimes \Gamma^\rdagger(\Psi)\Psi \rrangle \llangle \Phi| \\
& = & {\Gamma^\rdagger(\Psi)^\alpha}_\beta \Psi^{i \beta} \overline{\Phi_{j \alpha}} |\zeta_i \rangle \langle \zeta^j|
\end{eqnarray}

\begin{eqnarray}
\langle 1_\HS \Gamma^\ddagger(\Psi)\Phi|\Psi\rangle_* & = & \tr_\HA |\Psi \rrangle \llangle 1_\HS \otimes \Gamma^\ddagger(\Psi) \Phi| \\
& = & \Psi^{i \beta} \overline{{\Gamma^\ddagger(\Psi)_\beta}^\alpha} \overline{\Phi_{j \alpha}} |\zeta_i\rangle \langle \zeta^j| \\
& = & \Psi^{i \beta} {\Gamma^\rdagger(\Psi)^\alpha}_\beta \overline{\Phi_{j \alpha}} |\zeta_i\rangle \langle \zeta^j|
\end{eqnarray}
\end{proof}
It can be interesting to relate $\Ran W_\Psi$ (which appears in the definition of the $C^*$-adjoint) with the density operator.
\begin{prop}
Let $\Psi \in \HS \otimes \HA$, $W_\Psi \in \mathcal B(\HS)$ be such that $\PSI(W_\Psi) = \Psi$, and $\rho_\Psi = \|\Psi\|^2_* = \pi_{\mathcal A}(\Psi) = W_\Psi W_\Psi^\dagger \in \mathcal D(\HS)$. We have $\Ran W_\Psi = \Ran \rho_\Psi$ and $\coker W_\Psi = \ker \rho_\Psi$.
\end{prop}
\begin{proof}
$\coker W_\psi = \{\phi \in \HS, \Psi^{i \alpha} \overline \phi_i = 0\}$. $\phi \in \coker W_\Psi \Rightarrow \overline \phi_i \Psi^{i\alpha} \overline \Psi_{j\alpha} = 0 \Rightarrow \overline \phi_i {(\rho_\Psi)^i}_j = 0$ since $\rho_\Psi = \Psi^{i \alpha} \overline \Psi_{j \alpha} |\zeta_i \rangle \langle \zeta^j|$. We have then $\langle \phi|\rho_\Psi = 0 \Rightarrow \rho_\Psi |\phi \rangle = 0 \Rightarrow \phi \in \ker \rho_\Psi$. We have then $\coker W_\Psi \subset \ker \rho_\Psi$.\\
${(\rho_\Psi)^i}_j |\zeta_i \rangle = \Psi^{i\alpha} \overline \Psi_{j \alpha} |\zeta_i \rangle \Rightarrow \forall \phi \in \HS, {(\rho_\Psi)^i}_j \phi^j |\zeta_i \rangle = \Psi^{i\alpha} \overline \Psi_{j \alpha} \phi^j|\zeta_i \rangle$. We have then $\forall \phi \in \HS, \rho_\Psi \phi = W_\Psi (W_\Psi^\dagger \phi) \Rightarrow \rho_\Psi \phi \in \Ran W_\Psi$. We have then $\Ran \rho_\Psi \subset \Ran W_\Psi$.\\
Suppose that $(\zeta_i)_i$ the basis of $\HS$ is the eigenbasis of $\rho_\Psi$. ${(\rho_\Psi)^i}_j = \Psi^{i\alpha} \overline \Psi_{j\alpha} = p^i \delta^i_j$. We can write $\langle \Psi_j|\Psi^i \rangle_\HA = p^i \delta^i_j$ with $\Psi^i = \Psi^{i\alpha} \xi_\alpha \in \HA$. We have then $p^i = \|\Psi^i\|^2_\HA$. It follows that $p^i = 0 \Rightarrow \Psi^i = 0 \Rightarrow \Psi^{i\alpha} = 0, \forall \alpha$. $\Ran W_\Psi = \{\Psi^{i\alpha} \phi_\alpha |\zeta_i \rangle, \phi_\alpha \in \mathbb C\}$ and then $\Ran W_\Psi \subset \coRan \rho_\Psi$ ($\coRan \rho_\Psi = (\ker \rho_\Psi)^\bot$). But $\coRan \rho_\Psi = \Ran \rho_\Psi^\dagger = \Ran \rho_\Psi$ ($\rho_\Psi^\dagger = \rho_\Psi$). We have then $\Ran W_\Psi \subset \Ran \rho_\Psi$. \\
The two last paragraphs show that $\Ran W_\Psi = \Ran \rho_\Psi$. The first one shows that $\coker W_\Psi \subset \ker \rho_\Psi$ but $\dim \coker W_\Psi = \dim \HS - \dim \Ran W_\Psi$ (by definition $\HS = \Ran W_\Psi \overset{\bot}{\oplus} \coker W_\Psi$). We have then $\dim \coker W_\Psi = \dim \HS - \dim \Ran \rho_\Psi = \dim \ker \rho_\Psi$ (by virtue of the rank-nullity theorem). This induces that $\coker W_\Psi = \ker \rho_\Psi$.
\end{proof}

\section{Purification of Lindblad dynamics}
\subsection{From the Lindblad equation to the nonlinear Schr\"odinger equation}
\begin{theo}
Let $t\mapsto \rho(t) \in \mathcal D_0(\HS)$ be a solution of the Lindblad equation
\begin{equation}
\ihbar \dot \rho = [H_{\mathcal S},\rho] - \frac{\imath}{2} \gamma^k \{\Gamma^\dagger_k \Gamma_k,\rho\} + \imath \gamma^k \Gamma_k \rho \Gamma_k^\dagger
\end{equation}
with $H_{\mathcal S} \in \mathcal B(\HS)$ the system Hamiltonian, $\Gamma_k \in \mathcal B(\HS)$ the quantum jump operators and $\gamma^k \in [0,1]$ the relaxation rates. Let $t \mapsto \Psi_\rho(t) \in \HS \otimes \HA$ be a purification state of $\rho(t)$. $\Psi_\rho$ is solution of the following projected non-hermitian nonlinear Schr\"odinger equation:
\begin{eqnarray}
& & (1_\HS \otimes P_{\Ran \rho}) \ihbar \dot \Psi_\rho \nonumber \\
\label{PNLSE} & & = (1_\HS \otimes P_{\Ran \rho}) \left(H_{\mathcal S}\otimes 1_\HA \Psi_\rho - \frac{\imath}{2} \gamma^k \Gamma^\dagger_k \Gamma_k \otimes 1_\HA \Psi_\rho + \frac{\imath}{2} \gamma^k \Gamma_k \otimes \Gamma_k^\ddagger(\Psi_\rho) \Psi_\rho \right)
\end{eqnarray}
\end{theo}

\begin{proof}
Let $W_\rho(t) \in \mathcal B(\HS)$ be a standard purification of $\rho$, $\rho = W_\rho W_\rho^\dagger$, such that $\PSI(W_\rho) = \Psi_\rho$.
\begin{eqnarray}
& & \ihbar \dot \rho = \mathcal L(\rho) \\
& \Rightarrow & \ihbar \dot W_\rho W^\dagger_\rho + \ihbar W_\rho \dot W_\rho^\dagger = [H_{\mathcal S},\rho] - \frac{\imath}{2} \gamma^k \{\Gamma^\dagger_k \Gamma_k,\rho\} +2 \times \frac{\imath}{2} \gamma^k \Gamma_k \rho \Gamma_k^\dagger \\
  & \Rightarrow & \ihbar \dot W_\rho W_\rho^\dagger - (\ihbar \dot W_\rho W_\rho^\dagger)^\dagger = H_{\mathcal S} \rho -\frac{\imath}{2} \gamma^k \Gamma_k^\dagger \Gamma_k \rho + \frac{\imath}{2} \gamma^k \Gamma_k \rho \Gamma_k^\dagger \nonumber \\
  & & \qquad - \left(H_{\mathcal S} \rho -\frac{\imath}{2} \gamma^k \Gamma_k^\dagger \Gamma_k \rho + \frac{\imath}{2} \gamma^k \Gamma_k \rho \Gamma_k^\dagger\right)^\dagger \\
& \Rightarrow & \ihbar \dot W_\rho W_\rho^\dagger = H_{\mathcal S} \rho -\frac{\imath}{2} \gamma^k \Gamma_k^\dagger \Gamma_k \rho + \frac{\imath}{2} \gamma^k \Gamma_k \rho \Gamma_k^\dagger + K
\end{eqnarray}
where $K=K^\dagger$ is an arbitrary self-adjoint operator. We can set $K=0$ without loss of generality.
\begin{eqnarray}
& & \ihbar \dot W_\rho W_\rho^\dagger = H_{\mathcal S} W_\rho W_\rho^\dagger -\frac{\imath}{2} \gamma^k \Gamma_k^\dagger \Gamma_k W_\rho W_\rho^\dagger + \frac{\imath}{2} \gamma^k \Gamma_k W_\rho W_\rho^\dagger \Gamma_k^\dagger \\
  & \Rightarrow & \ihbar \dot W_\rho P_{\coRan W_\rho} = H_{\mathcal S} W_\rho P_{\coRan W_\rho}  -\frac{\imath}{2} \gamma^k \Gamma_k^\dagger \Gamma_k W_\rho  P_{\coRan W_\rho} \nonumber \\
  & & \qquad + \frac{\imath}{2} \gamma^k \Gamma_k W_\rho W_\rho^\dagger \Gamma_k^\dagger (W_\rho^\dagger)^\pinv
\end{eqnarray}
By application of $\PSI$ on this last equation we find
\begin{equation}
1_\HS \otimes P_{\Ran \rho} \ihbar \dot \Psi_\rho = H_{\mathcal S}\otimes P_{\Ran \rho} \Psi_\rho - \frac{\imath}{2} \gamma^k \Gamma^\dagger_k \Gamma_k \otimes P_{\Ran \rho} \Psi_\rho + \frac{\imath}{2} \gamma^k \Gamma_k \otimes (W_\rho^\dagger \Gamma_k^\dagger (W^\dagger_\rho)^\pinv)^\transp \Psi_\rho
\end{equation}
But $(W_\rho^\dagger \Gamma_k^\dagger (W^\dagger_\rho)^\pinv)^\transp = (W_\rho^\pinv \Gamma_k W_\rho)^{\dagger \transp} = \Gamma^\ddagger_k(\Psi_\rho)$.
\end{proof}

On the regular strata, no projection occurs and the purification state is solution of the nonlinear Schr\"odinger equation:
\begin{equation}
\ihbar \dot \Psi_\rho = \left(H_{\mathcal S} - \frac{\imath}{2} \gamma^k \Gamma^\dagger_k \Gamma_k\right) \otimes 1_\HA \Psi_\rho + \frac{\imath}{2} \gamma^k \Gamma_k \otimes \Gamma_k^\ddagger(\Psi_\rho) \Psi_\rho
\end{equation}

A comparison of this equation with other Sch\"odinger equations associated with open quantum systems can be found in \ref{othereq}.

\subsection{From the nonlinear Schr\"odinger equation to the Lindblad equation}
\begin{theo}
Let $\Psi \in \HS \otimes \HA$ be a solution of the nonlinear Schr\"odinger equation:
\begin{equation}
\label{NLSE}
\ihbar \dot \Psi = \left(H_{\mathcal S} - \frac{\imath}{2} \gamma^k \Gamma^\dagger_k \Gamma_k\right) \otimes 1_\HA \Psi + \frac{\imath}{2} \gamma^k \Gamma_k \otimes \Gamma_k^\ddagger(\Psi) \Psi
\end{equation}
Then $\rho_\Psi = \pi_{\mathcal A}(\Psi) = \|\Psi\|^2_* \in \mathcal D(\HS)$ is solution of the following master equation:
\begin{equation}
\label{PLE}
\ihbar \dot \rho_\Psi = [H_{\mathcal S},\rho_\Psi] - \frac{\imath}{2} \gamma^k \{\Gamma_k^\dagger \Gamma_k, \rho_\Psi\} + \frac{\imath}{2} \gamma^k \{\Gamma_k \rho_\Psi \Gamma_k^\dagger,P_{\Ran \rho_\Psi}\}
\end{equation}
\end{theo}

\begin{proof}
Let $H_{\mathcal U}(\bullet) = (H_{\mathcal S} - \frac{\imath}{2} \gamma^k \Gamma^\dagger_k \Gamma_k) \otimes 1_\HA + \frac{\imath}{2} \gamma^k \Gamma_k \otimes \Gamma_k^\ddagger(\bullet)$ be the nonlinear operator of the purified dynamics and $P_\Psi = |\Psi \rrangle \llangle \Psi|$.
\begin{eqnarray}
\ihbar \dot P_\Psi & = & |\ihbar \dot \Psi \rrangle \llangle \Psi| - |\Psi \rrangle \llangle \ihbar \dot \Psi| \\
& = & |H_{\mathcal U}(\Psi)\Psi \rrangle \llangle \Psi| - |\Psi \rrangle \llangle H_{\mathcal U}(\Psi)\Psi| \\
& = & [H_{\mathcal S}\otimes 1_\HA,P_\Psi] -\frac{\imath}{2}\gamma^k\{\Gamma_k^\dagger\Gamma_k \otimes 1_\HA,P_\Psi\} + \frac{\imath}{2}\gamma^k \Gamma_k \otimes \Gamma_k^\ddagger(\Psi) P_\Psi \nonumber \\
& & \qquad + \frac{\imath}{2}\gamma^k P_\Psi \Gamma_k \otimes \Gamma_k^\ddagger(\Psi)  
\end{eqnarray}
By applying the partial trace $\tr_\HA$ on this last equation, we find
\begin{equation}
\ihbar \dot \rho_\Psi = [H_{\mathcal S},\rho_\Psi] -\frac{\imath}{2}\gamma^k\{\Gamma_k^\dagger\Gamma_k,\rho_\Psi\} + \frac{\imath}{2}\gamma^k \langle \Psi|\Gamma_k \otimes \Gamma_k^\ddagger(\Psi) \Psi \rangle_* + \frac{\imath}{2} \gamma^k \langle \Gamma_k \otimes \Gamma_k^\ddagger(\Psi) \Psi|\Psi\rangle_*
\end{equation}
Since $\Gamma_k \otimes \Gamma_k^\ddagger(\Psi) = (1_\HS \otimes \Gamma_k^\ddagger(\Psi)) (\Gamma_k \otimes 1_\HA)$ we have
\begin{eqnarray}
  \ihbar \dot \rho_\Psi & = & [H_{\mathcal S},\rho_\Psi] -\frac{\imath}{2}\gamma^k\{\Gamma_k^\dagger\Gamma_k,\rho_\Psi\} + \frac{\imath}{2}\gamma^k \langle 1_\HS \otimes \Gamma_k^\rdagger(\Psi) \Psi|\Gamma_k \otimes 1_\HA \Psi \rangle_* \nonumber \\
  & & \qquad + \frac{\imath}{2} \gamma^k \langle \Gamma_k \otimes 1_\HA \Psi|  1_\HS \otimes \Gamma_k^\rdagger(\Psi)\Psi\rangle_* \\
& = & [H_{\mathcal S},\rho_\Psi] -\frac{\imath}{2}\gamma^k\{\Gamma_k^\dagger\Gamma_k,\rho_\Psi\} + \frac{\imath}{2}\gamma^k \langle P_{\Ran \rho_\Psi} \Gamma_k \otimes 1_\HA \Psi|\Gamma_k \otimes 1_\HA \Psi \rangle_* \nonumber \\
& & \qquad + \frac{\imath}{2} \gamma^k \langle \Gamma_k \otimes 1_\HA \Psi|  P_{\Ran \rho_\Psi} \Gamma_k \otimes 1_\HA \Psi\rangle_* \\
& = & [H_{\mathcal S},\rho_\Psi] -\frac{\imath}{2}\gamma^k\{\Gamma_k^\dagger\Gamma_k,\rho_\Psi\} +  \frac{\imath}{2}\gamma^k \Gamma_k \rho_\Psi \Gamma_k^\dagger P_{\Ran \rho_\Psi}  +  \frac{\imath}{2}\gamma^k P_{\Ran \rho_\Psi} \Gamma_k \rho_\Psi \Gamma_k^\dagger
\end{eqnarray}
\end{proof}
On the regular strata, $P_{\Ran \rho_\Psi} = 1_\HS$ and the master equation reduces to the Lindblad equation:
\begin{equation}
\ihbar \dot \rho_\Psi = [H_{\mathcal S},\rho_\Psi] - \frac{\imath}{2} \gamma^k \{\Gamma_k^\dagger \Gamma_k, \rho_\Psi\} + \imath \gamma^k \Gamma_k \rho_\Psi \Gamma_k^\dagger
\end{equation}

\section{Geometric phases of open quantum systems}
\subsection{The different notions of operator valued phases}
For the open quantum systems, the notion of phase cannot be the same that for the closed systems. In accordance with the $C^*$-module structure of the purification space $\HS \otimes \HA$, a phase for an open quantum system is certainly an operator in the $C^*$-algebra $\mathcal B(\HS)$ (we recall that a $C^*$-module mimics the structure of vector space by replacing the ring $\mathbb C$ by a $C^*$-algebra).
\begin{defi}
 We define two different notions of phase in the $C^*$-module for a state $\Psi \in \HS \otimes \HA$:
\begin{itemize}
\item We call phase by invariance, an operator $g\in GL(\HS)$ leaving invariant the $C^*$-norm:
\begin{equation}
\|g\Psi\|^2_* = \|\Psi\|^2_*
\end{equation}
\item We call phase by (unitary) equivariance, an operator $g\in GL(\HS)$ leaving equivariant the $C^*$-norm:
\begin{equation}
\|g\Psi\|^2_* = g\|\Psi\|^2_* g^{-1}
\end{equation}
\end{itemize}
\end{defi}
For an abelian $C^*$-algebra (as $\mathbb C$) the two notions are equivalent. We can note that $\forall g \in GL(\HS)$, $\forall \Psi \in \HS \otimes \HA$ we have
\begin{equation}
\|g\Psi\|^2_* = g\|\Psi\|^2_* g^\dagger
\end{equation}
and then each (SLOCC or isospectral) tranformation of $\mathcal S$ appears as a phase by non-unitary equivariance (a concept similar with the case of the non-hermitian quantum systems: due to the non-conservation of the norm during the dynamics, we consider ``non-unitary phases'' $g \in \mathbb C^*$ corresponding to a change of norm:$\|g\psi\|^2 = |g|^2 \|\psi\|^2$ with $\psi \in \HS$).\\
Let $\rho \in \mathcal D(\HS)$. The group of phases by invariance of $\pi^{-1}_{\mathcal A}(\rho)$ is the stabilizer (the isotropy subgroup) of $\rho$ for the left-right action: $GL(\HS)_\rho = \{g \in GL(\HS), g\rho g^\dagger = \rho \}$. The group of phases by equivariance is $U(\HS) \times GL(\HS)_\rho$. But since we will use only phases by invariance for isospectral transformations (inner to $\pi_D^{-1}(\sigma)$ with $\sigma = \pi_D(\rho)$), the group is reduced to $U(\HS)$. Because $\exists g_\rho \in GL(\HS)$ such that $\rho = g_\rho \sigma g_\rho^\dagger$ (with $\sigma = \pi_D(\rho)$, $g_\rho = s_\rho u_\rho$ with $u_\rho \in U(\HS)$ an unitary transformation and $s_\rho \in GL(\HS)/U(\HS)$ a SLOCC transformation), we have $GL(\HS)_\rho = GL(\HS)_{g_\rho \sigma g_\rho^\dagger} = g_\rho GL(\HS)_\sigma g_\rho^{-1}$. The isotropy subgroups are then all isomorphic to the same model group for all $\rho$ chosen in a single stratum (but it is different between two strata).\\
To simplify the notations, we denote by $K = U(\HS) \simeq U(n)$ (or $U(\HA)$ if it acts on the right) the group of the phases by equivariance (or the group of ancilla transformations), and by $G = GL(\HS) \simeq GL(n,\mathbb C)$ the group of the non-unitary phases. For a stratum $\Sigma^i(\HS)$ we denote by $H^i = GL(\HS)_\sigma \simeq GL(q_1)^{k_1} \times ... \times GL(q_l)^{k_l}$ the stabilizer of the diagonal density operators with sorted eigenvalues ($k_i$ is the number of eigenvalues with degeneracy equal to $q_i$). The group of all phases by invariance of the stratum $\Sigma^i(\HS)$ is the normal closure of $H^i$: $\overline H^i = \bigcup_{g \in G} gH^ig^{-1} $ ($\overline H^i$ is the smallest normal subgroup of $G$ including $H^i$).\\
We can also say that $K$ is the group of the inobservable reconfigurations of the ancilla, $G$ is the group of the unitary and SLOCC transformations of the system performed by the physicist, and $\overline{H^i}$ is the subgroup of inefficient transformations of the system (a transformation is inefficient with respect to a particular density operator).\\

The non-unitary operator valued phases have been defined only with respect to the norm of the $C^*$-module (the projection of the purified states onto the density operators). But the non-commutativity of the $C^*$-algebra $\mathcal B(\HS)$ and the nonlinearity of the Hamiltonian of the purified dynamics ($H_{\mathcal U}(\bullet) = (H_{\mathcal S} - \frac{\imath}{2} \gamma^k \Gamma^\dagger_k \Gamma_k) \otimes 1_\HA + \frac{\imath}{2} \gamma^k \Gamma_k \otimes \Gamma_k^\ddagger(\bullet)$) induce a difficulty, because the phases do not commute with the generator of the dynamics. 
\begin{defi}
We call phase with respect to the generator of the dynamics, an operator $g\in G$ such that
\begin{equation}
H_{\mathcal U}(g\Psi) g\Psi = gH_{\mathcal U}(\Psi) \Psi
\end{equation}
\end{defi}
For closed quantum systems where the Hamiltonian is linear and the phases are scalars, this condition is trivial.
\begin{propo}
The group of the phases with respect to the generator of the dynamics is $G_{\mathcal L} = G_{H^{eff}} \cap \bigcap_k G_{\Gamma_k}$ ($H^{eff} = H_{\mathcal S} - \frac{\imath}{2} \gamma^k \Gamma_k^\dagger \Gamma_k$), where the isotropy subgroups are defined for the adjoint action ($G_{H^{eff}} = \{g \in G, g^{-1} H^{eff} g = H^{eff}$).
\end{propo}
\begin{proof}
Let $W_\Psi \in \mathcal B(\HS)$ be such that $\PSI(W_\Psi) = \Psi$.
\begin{eqnarray}
H_{\mathcal U}(g\Psi)g\Psi & = & H^{eff} \otimes 1_\HA g\Psi + \frac{\imath}{2} \gamma^k \Gamma_k \otimes \Gamma^{\ddagger}_k(g\Psi) g\Psi \\
& = & H^{eff}g \otimes 1_\HA \Psi + \frac{\imath}{2} \gamma^k \Gamma_kg \otimes \left((gW_\Psi)^\pinv\Gamma_k gW_\Psi\right)^{\transp\dagger} \Psi \\
& = & H^{eff}g \otimes 1_\HA \Psi + \frac{\imath}{2} \gamma^k \Gamma_kg \otimes \left((W_\Psi^\pinv g^{-1} \Gamma_k gW_\Psi\right)^{\transp\dagger} \Psi
\end{eqnarray}
$H^{eff}g = gH^{eff}$ if $g \in G_{H^{eff}}$, $\Gamma_k g = g \Gamma_k$ if $g \in G_{\Gamma_k}$ and $g^{-1} \Gamma_k g = \Gamma_k$ if $g\in G_{\Gamma_k}$. $H_{\mathcal U}(g\Psi)g\Psi = gH_{\mathcal U}(\Psi)\Psi$ if $g \in G_{H^{eff}} \cap \bigcap_k G_{\Gamma_k}$.
\end{proof}
We remark that $\forall k \in K$, $H_{\mathcal U}(1_\HS \otimes k^\transp \Psi) 1_\HS \otimes k^\transp \Psi = 1_\HS \otimes k^\transp H_{\mathcal U}(\Psi) \Psi$ (because $\Gamma^{\ddagger}_k(1_\HS \otimes k^\transp \Psi) =  (k^{-1} W_\Psi^\pinv \Gamma_k W_\Psi k)^{\transp\dagger} = k^\transp \Gamma^{\ddagger}_k(\Psi) (k^\transp)^{-1}$).\\
We remark moreover that if the system Hamiltonian $H_{\mathcal S}$ and/or the jump operators $\Gamma_k$ are time-dependent, $G_{\mathcal L}$ is time-dependent. In that case, a possibility to avoid difficulties is to consider $\mathring G_{\mathcal L} = \bigcap_t G_{\mathcal L}(t)$ but this group can be reduced to the unitary center of $\mathcal B(\HS)$ ($U(1)$ in finite dimension) except if $\forall t$, $H_{\mathcal S}(t)$ and $\Gamma_k(t)$ belong to a same (small) subalgebra of $\mathcal B(\HS)$.

\subsection{Operator valued geometric phases}
For closed quantum systems, the geometric phase concept is related to the cyclicity of the projected dynamics. Let $[0,T] \ni t \mapsto \rho(t) \in \mathcal D_0(\HS)$ be a density operator solution of the Lindblad equation $\ihbar \dot \rho = \mathcal L(\rho)$. We say that the projected dynamics is cyclic if $\pi_D(\rho(T)) = \pi_D(\rho(0))$ (the spectrum of the statistical probabilities is the same at the start and at the end of the dynamics). We search a density operator $[0,T] \ni t \mapsto \tilde \rho(t) \in \mathcal D(\HS)$ such that
\begin{eqnarray}
\tilde \rho(T) & = & \tilde \rho(0) \\
\forall t, \exists g(t) \in G, \quad \rho(t) & = & g(t)\tilde \rho(t) g(t)^\dagger
\end{eqnarray}
$\tilde \rho$ is the cyclic density operator associated with the density operator of cyclic projection. We note that $g(t) \in G$ and not $K$ because there is no reason for which the transformation of $\rho$ into a cyclic density operator implies no SLOCC operations.
\begin{theo} \label{thgeoph}
A density operator $\rho(t)$ with cyclic projection and its cyclic density operator $\tilde \rho(t)$ are related by $\rho(t) = g(t)\tilde \rho(t) g(t)^\dagger$ with the non-unitary operator valued phase $g(t)$ defined by
\begin{equation}
g(t) = \Teg^{-\ihbar^{-1} \int_0^t E(g(t') \tilde \rho(t') g(t')^\dagger)dt'} \Ted^{- \int_0^t (i_{\dot W_{\tilde \rho}(t')} \mathfrak A+\eta(t'))dt'}
\end{equation}
where the dynamical phase generator is defined by
\begin{equation}
\forall \rho \in \mathcal D(\HS), \quad E(\rho) = H_{\mathcal S} - \frac{\imath}{2} \gamma^k \Gamma_k^\dagger \Gamma_k + \frac{\imath}{2} \gamma^k \Gamma_k \rho \Gamma_k^\dagger \rho^\pinv \in \mathcal B(\HS)
\end{equation}
 and the geometric phase generator of the first kind is defined by
\begin{equation}
\mathfrak A = dW W^\pinv \in \Omega^1(\mathcal B(\HS),\mathcal B(\HS))
\end{equation}
$i_{\dot W_{\tilde \rho}}$ is the inner product with $\dot W_{\tilde \rho}$ viewed as a tangent vector, $W_{\tilde \rho}(t)$ being a standard purification of $\tilde \rho(t)$: $i_{\dot W_{\tilde \rho}} \mathfrak A = \dot W_{\tilde \rho}(t) W_{\tilde \rho}(t)^\pinv$. The geometric phase generator of the second kind is defined as
\begin{equation}
\eta(t) = W_{\tilde \rho}(t) \dot k(t) k(t)^{-1} W_{\tilde \rho}(t)^\pinv + h(t) \in \overline{\mathfrak h^i}
\end{equation}
with $k(t) \in K$ and $h(t) \in \mathcal B(\ker \tilde \rho(t))$ ($\overline{\mathfrak h^i}$ is the Lie algebra of $\overline {H^i}$ with $i$ the index of the stratum of $\tilde \rho$).
\end{theo}
We can note that
\begin{equation}
\mathfrak A \|\PSI(W)\|^2_* = \langle \PSI(W) | d\PSI(W) \rangle_* = \tr_\HA |d\PSI(W) \rrangle \llangle \PSI(W)|
\end{equation}
This definition is similar to the definition of the non-unitary geometric phase generator of non-hermitian closed quantum systems (with the $C^*$-module in place of the Hilbert space of the system). It is moreover the non-adiabatic generalisation of the $C^*$-geometric phase introduced in ref.\cite{Viennot1,Viennot2,Viennot3,Viennot4}. We can also show that
\begin{equation}
d\rho = \mathfrak A \rho + \rho \mathfrak A^\dagger
\end{equation}
The definition of the dynamical phase $\Teg^{-\ihbar \int_0^t E(g\tilde \rho g^\dagger)dt}$ induces that the definition of total phase $g(t)$ is implicit (the expression of $g$ depends on itself). This is a consequence of the nonlinearity of the Schr\"odinger equation of the purified dynamics (we can say that the dynamical phase is a ``nonlinear dynamical phase'').\\
We can remark that $h(t)=0$ if the cyclic dynamics ($t \mapsto \tilde \rho(t)$) takes place in regular strata.\\
\begin{proof}
Since $\rho = g \tilde \rho g^\dagger$, $\exists k \in K$ such that $W_\rho = gW_{\tilde \rho}k$. The nonlinear Schr\"odinger equation (for a standard purification state) is
\begin{eqnarray}
\ihbar \dot W_\rho & = & H^{eff} W_\rho + \frac{\imath}{2} \gamma^k \Gamma_k W_\rho \Gamma_k^{\ddagger}(\PSI(W_\rho))^\transp \\
& = & H^{eff} W_\rho + \frac{\imath}{2} \gamma^k \Gamma_k W_\rho (W_\rho^\pinv \Gamma_k W_{\rho})^\dagger \\
& = & H^{eff} W_\rho + \frac{\imath}{2} \gamma^k \Gamma_k W_\rho W_\rho^\dagger \Gamma_k^\dagger (W_{\rho}^\pinv)^\dagger
\end{eqnarray}
with $H^{eff} = H_{\mathcal S} - \frac{\imath}{2} \gamma^k \Gamma_k^\dagger \Gamma_k$. It follows that
\begin{equation}
\dot g W_{\tilde \rho} k + g \dot W_{\tilde \rho} k + g W_{\tilde \rho} \dot k = -\ihbar^{-1} \left(H^{eff} g W_{\tilde \rho} k +  \frac{\imath}{2} \gamma^l \Gamma_l  g W_{\tilde \rho} W_{\tilde \rho}^\dagger g^\dagger \Gamma_l^\dagger (g^\dagger)^{-1} (W_{\tilde \rho}^\pinv)^\dagger k \right)
\end{equation}
and then
\begin{equation}
g^{-1}\dot g W_{\tilde \rho} + W_{\tilde \rho} \dot k k^{-1} = - \dot W_{\tilde \rho} -\ihbar^{-1} \left(g^{-1} H^{eff} g W_{\tilde \rho} + \frac{\imath}{2} \gamma^l g^{-1} \Gamma_l g W_{\tilde \rho} W_{\tilde \rho}^\dagger g^\dagger \Gamma_l^\dagger (g^\dagger)^{-1} (W_{\tilde \rho}^\pinv)^\dagger \right)
\end{equation}
By multiplying on the right this expression by $W_{\tilde \rho}^\pinv$ we find
\begin{equation}
g^{-1} \dot g + W_{\tilde \rho} \dot k k^{-1} W_{\tilde \rho}^\pinv = - \dot W_{\tilde \rho} W_{\tilde \rho}^\pinv - \ihbar^{-1} g^{-1} E(g\tilde \rho g^\dagger)g - h
\end{equation}
with $h \in \mathcal B(\ker \tilde \rho)$. We set $\eta = W_{\tilde \rho} \dot k k^{-1} W_{\tilde \rho}^\pinv + h$. Let $g_{\mathfrak A}$ be such that $g = \Teg^{-\ihbar^{-1} \int_0^t E(g \tilde \rho g^\dagger)dt} g_{\mathfrak A}$.
\begin{eqnarray}
\dot g & = & - \ihbar^{-1} E(g \tilde \rho g^\dagger) g + \Teg^{-\ihbar^{-1} \int_0^t E(g \tilde \rho g^\dagger)dt} \dot g_{\mathfrak A} \\
g^{-1} \dot g & = & - \ihbar^{-1} g^{-1} E(g \tilde \rho g^\dagger) g + g_{\mathfrak A}^{-1} \dot g_{\mathfrak A}
\end{eqnarray}
It follows that
\begin{equation}
g_{\mathfrak A}^{-1} \dot g_{\mathfrak A} = - \dot W_{\tilde \rho} W_{\tilde \rho} - \eta \Rightarrow g_{\mathfrak A} = \Ted^{-\int_0^t (\dot W_{\tilde \rho} W_{\tilde \rho} +\eta)dt}
\end{equation}
One needs now only to show that $\eta \in \overline{\mathfrak h^i}$:
\begin{eqnarray}
\eta \tilde \rho + \tilde \rho \eta^\dagger & = & W_{\tilde \rho} \dot k k^{-1} W_{\tilde \rho}^\dagger + W_{\tilde \rho} k \frac{dk^{-1}}{dt} W_{\tilde \rho}^\dagger \\
& = & W_{\tilde \rho} \dot k k^{-1} W_{\tilde \rho}^\dagger - W_{\tilde \rho} \dot k k^{-1} W_{\tilde \rho}^\dagger\\
& = & 0
\end{eqnarray}
since $\tilde \rho = W_{\tilde \rho} W_{\tilde \rho}^\dagger$, $\frac{dk^{-1}}{dt} = -k^{-1}\dot k k^{-1}$ and $h\tilde \rho=0$. Because of $g_{\tilde \rho}^{-1} H^i g_{\tilde \rho} = GL(\HS)_{\tilde \rho} = \{g \in GL(\HS), g\tilde \rho g^\dagger = \tilde \rho\}$ ($\tilde \rho = g_{\tilde \rho} \pi_D(\tilde \rho) g_{\tilde \rho}^\dagger$), the Lie algebra of $GL(\HS)_{\tilde \rho}$ is $\mathfrak{gl}(\HS)_{\tilde \rho} =\{X \in \mathcal B(\HS), X \tilde \rho + \tilde \rho X^\dagger = 0\}$. We have then $\eta \in \mathfrak{gl}(\HS)_{\tilde \rho} \subset \overline{\mathfrak h^i}$.
\end{proof}

From the point of view of the purified dynamics, we have
\begin{equation}
\Psi_\rho(t) = \Teg^{-\ihbar^{-1} \int_0^t E(g\tilde \rho g^\dagger)dt} \Ted^{-\int_0^t (i_{\dot W_{\tilde \rho}} \mathfrak A + \eta) dt} \otimes k^\transp(t) \Psi_{\tilde \rho}(t)
\end{equation}
with $\Psi_\rho = \PSI(W_\rho)$, and $i_{\dot W_{\tilde \rho}} \mathfrak A \|\Psi_{\tilde \rho}\|^2_*= \langle \Psi_{\tilde \rho} | \dot \Psi_{\tilde \rho} \rangle_*$.\\

The right geometric phase $k(t) \in K$ is arbitrary in the sense that $\pi_{HS}(g\tilde Wk) = \underbrace{g \tilde \rho g^\dagger}_{\rho} = \pi_{HS}(g \tilde W)$, but it induces the geometric phase generator of the second kind $\eta_k = W_{\tilde \rho} \dot k k^{-1} W_{\tilde \rho}^\pinv$ in the left geometric phase. Let $A^R \in \Omega^1(\mathcal B(\HS),\mathcal B(\HS))$ be the generator of the right geometric phase:
\begin{equation}
k(t) = \Teg^{-\int_0^t i_{\dot W_{\tilde \rho}(t')} A^R dt'}
\end{equation}
A possible choice consists to use the definition of the Uhlmann geometric phase \cite{Uhlmann1,Uhlmann2}:
\begin{equation}
dW = {^{Uhl}}\mathfrak A W + W A^R
\end{equation}
where ${^{Uhl} \mathfrak A}$ is solution of the following equation:
\begin{equation}
d\rho = {^{Uhl}}\mathfrak A \rho + \rho {^{Uhl}}\mathfrak A \qquad {^{Uhl}} \mathfrak A^\dagger = {^{Uhl}} \mathfrak A
\end{equation}
This choice is set in order to $W_{\parallel} = W_{\tilde \rho} k$ satisfies $W_{\parallel}^\dagger \dot W_{\parallel} = \dot W_{\parallel}^\dagger W_{\parallel} \iff \dot W_{\parallel} W_{\parallel}^\pinv = W_{\parallel}^{\pinv \dagger} \dot W_{\parallel}^\dagger$, which is the Uhlmann's definition of the parallel transport for the density matrices \cite{Uhlmann1,Uhlmann2}. With this choice of right geometric phase, we have:
\begin{eqnarray}
i_{\dot W_{\tilde \rho}} \mathfrak A + W_{\tilde \rho} \dot k k^{-1} W_{\tilde \rho}^\pinv & = & i_{\dot W_{\tilde \rho}}\left(\mathfrak A + W A^R W^\pinv \right) \\
& = & i_{\dot W_{\tilde \rho}}\left( dWW^\pinv + W A^R W^\pinv \right) \\
& = & i_{\dot W_{\tilde \rho}}\left({^{Uhl}}\mathfrak A \right)
\end{eqnarray}
The left generator appearing in the definition of the Uhlmann (right) geometric phase, generates then the total left geometric phase:
\begin{eqnarray}
\Ted^{-\int_0^T (i_{\dot W_{\tilde \rho}(t)} \mathfrak A + \eta_k(t))dt} & = & \Ted^{-\int_0^T i_{\dot W_{\tilde \rho}(t)} {^{Uhl}}\mathfrak A dt} \\
& = & \Ped^{- \oint_{\mathcal C} {^{Uhl}}\mathfrak A}
\end{eqnarray}
where $\mathcal C$ is the closed curve parametrized by $[0,T] \ni t \mapsto W_{\tilde \rho}(t)$ in $\mathcal B(\HS) \simeq \mathfrak M_{n \times n}(\mathbb C)$ viewed as a manifold ($\Ped$ denoting the path-anti-ordered exponential).\\
Since the total left geometric phase is automatically adaptated to the choice of the right geometric phase (by the presence of the left geometric phase generator of the second kind $\eta_k = W_{\tilde \rho} \dot k k^{-1} W_{\tilde \rho}^\pinv$), no natural choice of right geometric phase is imposed by the nonlinear Schr\"odinger equation and its associated Lindblad equation. Other definitions of the right geometric phase corresponding to other Uhlmann like connections can be choosen. This is the reason for which Dittman, Uhlmann and Rudolph \cite{Dittmann,Dittmann2} have find a large class of connections defining right geometric phases and compatible with the density matrix theory.

\subsection{Special cases of left geometric phases}
\paragraph{Geometric phase with respect to the generator of the dynamics:} We suppose that the dynamics takes place in regular strata. Let $\mathfrak g_{\mathcal L}$ be the Lie algebra of $G_{\mathcal L}$ and let $\forall \Psi \in \HS \otimes \HA$, $S_{\mathcal L}(\Psi) = \mathfrak g_{\mathcal L} \Psi$ be the Hilbert subspace which is the orbit of $\Psi$ by $\mathfrak g_{\mathcal L}$. Let $\{X_i\}_i$ be a set of generators of the Lie algebra $\mathfrak g_{\mathcal L}$ such that $\tr(X^{i\dagger} X_j) = \delta^i_j$. By construction $\{\PSI(X_iW)\}_i$ constitutes a basis of $S_{\mathcal L}(\PSI(W))$ (with $W \in \mathcal B(\HS)$). In spite of the orthonormalisation of the generators of $\mathfrak g_{\mathcal L}$, the basis $\{\PSI(X_iW)\}_i$ is not orthonormal. Let $\{\PSI(X_i (W^\dagger)^{-1})\}_i$ be the associated biorthonormal basis:
\begin{eqnarray}
\llangle \PSI(X^i (W^{\dagger})^{-1})|\PSI(X_j W) \rrangle & = & \langle X^i (W^{\dagger})^{-1}| X_j W \rangle_{HS} \\
& = & \tr\left(W^{-1} X^{i\dagger} X_j W \right) \\
& = & \tr\left(WW^{-1} X^{i\dagger} X_j \right) \\
& = & \tr(X^{i\dagger} X_j) \\
& = & \delta^i_j
\end{eqnarray}
(where we have used the cyclicity of the trace). Let $\PSI(P_{\mathcal L}(W))$ be the (non-orthogonal) projection onto $S_{\mathcal L}(\PSI(W))$ in $\HS \otimes \HA$ defined by
\begin{equation}
P_{\mathcal L}(W) = |X_i W \rangle_{HS}\langle X^i (W^{\dagger})^{-1}| = X_i W \tr\left(W^{-1} X^{i\dagger} \bullet \right)
\end{equation}
Let $A^L \in \Omega^1(\mathcal B(\HS),\mathfrak g_{\mathcal L})$ be defined as
\begin{eqnarray}
A^L \|\PSI(W)\|^2_* & = & \langle \PSI(W)| \PSI(P_{\mathcal L}(W)) [d\PSI(W)] \rangle_* \\
A^L WW^\dagger & = & P_{\mathcal L}(W)[dW] W^\dagger \\
A^L WW^\dagger & = & \tr\left(W^{-1} X^{i\dagger} dW \right) X_i WW^\dagger \\
A^L & = & \tr\left(X^{i\dagger} dW W^{-1}\right) X_i \\
A^L & = & \tr\left(X^{i\dagger} \mathfrak A \right) X_i
\end{eqnarray}
$A^L$ is the part of $\mathfrak A$ which induces a phase with respect to the generator of the dynamics.
\begin{propo}
$i_{\dot W_{\tilde \rho}} \mathfrak A = i_{\dot W_{\tilde \rho}} A^L$ (the geometric phase is a phase with respect to the generator of the dynamics) if and only if $E(\rho) \in \mathfrak g_{\mathcal L}$:
\begin{eqnarray}
\left[E(\rho),H^{eff}\right] & = & 0 \\
\left[E(\rho),\Gamma_k\right] & = & 0 \\
\left[E(\rho),\Gamma_k^\dagger\right] & = & 0
\end{eqnarray}
\end{propo}

\begin{proof}
$i_{\dot W_{\tilde \rho}} \mathfrak A = \dot W_{\tilde \rho} W_{\tilde \rho}$ and $\dot W_{\tilde \rho} = -g^{-1}g W_{\tilde \rho} + g^{-1} \dot W_\rho k^{-1} - W_{\tilde \rho} \dot kk^{-1}$. We have then $i_{\dot W_{\tilde \rho}} \mathfrak A = -g^{-1} g - \ihbar^{-1} g^{-1} E(\rho) g - \eta_k$. $i_{\dot W_{\tilde \rho}} \mathfrak A \in \mathfrak g_{\mathcal L}$ if $g \in G_{\mathcal L}$ and $E(\rho) \in \mathfrak g_{\mathcal L}$.
\end{proof}
Note that $g \in G_{\mathcal L} \Rightarrow E(g\tilde \rho g^\dagger) = gE(\tilde \rho) g^{-1}$. It follows that
\begin{eqnarray}
g(t) & = & \Teg^{-\ihbar^{-1} \int_0^t g(t') E(\tilde \rho(t')) g(t')^{-1} dt'} \Ted^{-\int_0^t (i_{\dot W_{\tilde \rho(t')}} A^L+\eta_k) dt'} \\
& = & \Ted^{-\ihbar^{-1} \int_0^t E(\tilde \rho(t')) dt' - \int_0^t i_{\dot W_{\tilde \rho(t')}} (A^L+\eta_k) dt'}
\end{eqnarray}
In that case, the nonlinearity of the dynamical phase does not then occur but the dynamical and the geometric phases are not seperated.\\

\paragraph{Geometric phase by invariance:} We suppose that the dynamics is inner to a regular stratum $\Sigma^i(\HS)$. Let $H^i \simeq GL(q_1) \times ... \times GL(q_l)$ be the stabilizer of the stratum $\Sigma^i(\HS)$ (we do not necessarily consider that $q_j \not= q_k$ for $j\not=k$). Let $\mathfrak h^i \simeq \mathfrak u(q_1) \oplus ... \oplus \mathfrak u(q_l)$ be its Lie algebra. $\forall \sigma \in \Sigma^i(\HS)$, let $\mathcal P^\sigma : \mathcal B(\HS) \to \mathfrak h^i$ be defined by
\begin{equation}
\mathcal P^\sigma(X) = P^{j\sigma} X P_j^\sigma
\end{equation}
where $P_j^\sigma$ is the orthogonal projection onto the $q_j$-dimensional eigenspace of $\sigma$ associated with the eigenvalue $p_j$ (which has a degeneracy equal to $q_j$). Let $\rho \in \mathcal D(\HS)$ be such that $\pi_D(\rho) = \sigma$. We know that $\exists g_\rho \in G$ such that $\rho = g_\rho \sigma g_\rho^\dagger$ and $GL(\HS)_\rho = g_\rho H^i g_\rho^{-1}$. Let $\mathcal P^\rho: \mathcal B(\HS) \to \mathfrak{gl}(\HS)_\rho$ be defined by $\mathcal P^\rho(X) = P^{j\rho} X P_j^\rho$ with $P_j^\rho = g_\rho P_j^\sigma g_\rho^{-1}$.\\
Let $A^{BL} \in \Omega^1(\mathcal B(\HS),\overline{\mathfrak h^i})$ be defined as
\begin{equation}
A^{BL} = \mathcal P^{WW^\dagger}(\mathfrak A) = P^{j WW^\dagger} dWW^{-1} P_j^{WW^\dagger}
\end{equation}
$A^{BL}$ is the part of $\mathfrak A$ which induces a phase by invariance.

\begin{propo}
$i_{\dot W_{\tilde \rho}} \mathfrak A = i_{\dot W_{\tilde \rho}} A^{BL}$ (the geometric phase is a phase by invariance) if and only if $E_{g_\rho}(\sigma) \in \mathfrak h^i$ with $E_{g_\rho}(\sigma) = g_\rho^{-1} E(g_\rho \sigma g_\rho^{-1}) g_\rho$ ($g_\rho \in U(\HS)$, $\rho(t)=g_\rho(t)\sigma(t)g_\rho(t)^{-1} \in \mathcal D_0(\HS)$).
\end{propo}

\begin{proof}
 $i_{\dot W_{\tilde \rho}} \mathfrak A = -g^{-1} \dot g - \ihbar^{-1} g^{-1} E(\rho) g - \eta_k$ and then $i_{\dot W_{\tilde \rho}} \mathfrak A \in \mathfrak{gl}(\HS)_\rho$ if $g \in GL(\HS)_\rho$ and $E(\rho) \in \mathfrak{gl}(\HS)_\rho \iff g_\rho^{-1} E(\rho) g_\rho \in \mathfrak h^i$.
\end{proof}

The geometric phase and the dynamical phase are phases by invariance, i.e. $\rho(T) = \rho(0)$ (the density operator is cyclic), if
\begin{equation}
g_{\rho}^{-1} H_{\mathcal S} g_\rho - \frac{\imath}{2}\gamma^k g_\rho^{-1} \Gamma_k^\dagger \Gamma_k g_\rho + \frac{\imath}{2} \gamma^k g_\rho^{-1} \Gamma_k g_\rho \sigma g_\rho^{-1} \Gamma_k^\dagger g_\rho \sigma^{-1} g_\rho^{-1} \in \mathfrak h^i
\end{equation}

The elements of $\mathfrak{gl}(\HS)_\rho$ acting on the left can be converted as elements acting on the right:
\begin{eqnarray}
X \in \mathfrak{gl}(\HS)_\rho & \Rightarrow & X\rho+\rho X^\dagger = 0 \\
& \Rightarrow & XWW^\dagger + WW^\dagger X^\dagger = 0 \\
& \Rightarrow & XW + WW^\dagger X^\dagger (W^\dagger)^{-1} = 0 \\
& \Rightarrow & XW = - W(W^{-1} XW)^\dagger
\end{eqnarray}
The left action of $X \in \mathfrak{gl}(\HS)_\rho$ is then equivalent to the right action of $-(W^{-1} XW)^\dagger$. We have then
\begin{equation}
\label{defgeosjoqvist}
\mathfrak A W = (\mathfrak A-A^{BL})W + W A^{BR}
\end{equation}
with
\begin{eqnarray}
A^{BR} & = & - (W^{-1} A^{BL} W)^\dagger \\
& = & - W^\dagger \left(P^{jWW^\dagger} \right)^\dagger (W^{-1})^\dagger dW^\dagger \left(P_j^{WW^\dagger} \right)^\dagger (W^{-1})^\dagger
\end{eqnarray}
$P_j^{WW^\dagger} = g_{WW^\dagger} P_j^\sigma g_{WW^\dagger}^{-1}$ and $W = g_{WW^\dagger} \sqrt{\sigma}$, we have then $W^{-1} P_j^{WW^\dagger} W = \sqrt{\sigma}^{-1} P_j^\sigma \sqrt{\sigma} = P_j^\sigma$ since $P_j^\sigma \sqrt{\sigma} = \sqrt{\sigma} P_j^\sigma$ ($P_j^\sigma$ is an eigenprojection of $\sigma$). We have then
\begin{equation}
A^{BR} = - P^{j\sigma} dW^\dagger (W^\dagger)^{-1} P^\sigma_j
\end{equation}
Let $W \in S_{HS}\mathcal B(\HS)$. $dW^\dagger (W^\dagger)^{-1} = d\sqrt{\sigma} \sqrt{\sigma}^{-1} + \sqrt{\sigma} dg_{WW^\dagger}^{-1} g_{WW^\dagger} \sqrt{\sigma}^{-1}$. Since $d\sigma = d\sqrt \sigma \sqrt \sigma + \sqrt \sigma d\sqrt \sigma \Rightarrow d\sqrt \sigma \sqrt \sigma^{-1} = d\sigma \sigma^{-1} - \sqrt \sigma d\sqrt \sigma \sigma^{-1}$ we have
\begin{equation}
dW^\dagger (W^\dagger)^{-1} = d\sigma \sigma^{-1} - \sqrt \sigma d\sqrt \sigma \sigma^{-1} - \sqrt{\sigma} g_{WW^\dagger}^{-1} d g_{WW^\dagger} \sqrt{\sigma}^{-1} \qquad (W \in S_{HS}\mathcal B(\HS))
\end{equation}
But $W^\dagger dW \sigma^{-1} = \sqrt \sigma g_{WW^\dagger}^{-1} dg_{WW^\dagger} \sqrt{\sigma}^{-1} + \sqrt \sigma d\sqrt \sigma \sigma^{-1}$, and then
\begin{equation}
dW^\dagger (W^\dagger)^{-1} = d\sigma \sigma^{-1} - W^\dagger dW \sigma^{-1} \qquad (W \in S_{HS}\mathcal B(\HS))
\end{equation}
Finally, the restriction of $A^{BR}$ onto $S_{HS}\mathcal B(\HS)$ is
\begin{equation}
A^{BR}_{|S_{HS}\mathcal B(\HS)} = P^{j\sigma} W^\dagger dW P^\sigma_j \sigma^{-1} - P^{j\sigma} d\sigma \sigma^{-1} P_j^\sigma
\end{equation}
This is the expression of the generator of the Sj\"oqvist-Andersson geometric phase\cite{Sjoqvist,Tong,Andersson} (note that in the works of Sj\"oqvist and Andersson, $\dot\sigma \sigma^{-1} = 0$ since they consider only isospectral dynamics). $A^{BR} =  - P^{j\sigma} dW^\dagger (W^\dagger)^{-1} P^\sigma_j$ is then the non-unitary generalization of the generator of the Sj\"oqvist-Andersson geometric phase. We can note that equation (\ref{defgeosjoqvist}) can rewritten as
\begin{equation}
dW = (\mathfrak A - A^{BL}) W + WA^{BR}
\end{equation}
which is an equation defining an Uhlmann like connection. In other words, we can choose the right geometric phase as being the Sj\"oqvist geometric phase in place of the Uhlmann geometric phase ($\dot k k^{-1} = - A^{BR}$). In that case, the total left geometric phase is such that
\begin{equation}
\Ted^{-\int_0^T (i_{\dot W_{\tilde \rho}(t)} \mathfrak A + \eta_k(t)) dt} \in G/G_\rho
\end{equation}

We will say that this Sj\"oqvist geometric phase is two-sided in the sense that it can be considered on the left ($A^{BL}$) as a inefficient transformation of the system, or it can be considered on the right ($A^{BR}$) as a reconfiguration of the ancilla. Its invariance nature implies that it can be revealed only by interferometry in accordance with its discovery by Sj\"oqvist etal \cite{Sjoqvist,Tong}.\\

\paragraph{Adiabatic geometric phase :} Suppose that $t \mapsto H_{\mathcal S}(t)$ and $t \mapsto \Gamma_k(t)$ are continuous and cyclic operators with respect to the time: $H_{\mathcal S}(T) = H_{\mathcal S}(0)$ and $\Gamma_k(T) = \Gamma_k(0)$. Moreover suppose that $|\gamma| = \max_k |\gamma^k| \ll 1$. The nonlinear Schr\"odinger equation for the purified dynamics can be written as
\begin{equation}
\ihbar \dot \Psi = H^{eff} \otimes 1_\HA \Psi + \frac{\imath}{2} \gamma^k \Gamma_k \otimes (W^{(0) \dagger} \Gamma_k^\dagger W^{(0)\dagger \pinv})^\transp \Psi + \mathcal O(|\gamma|^2)
\end{equation}
with $\Psi^{(0)} = \PSI(W^{(0)})$ be such that:
\begin{equation}
\ihbar \dot \Psi^{(0)} = H_{\mathcal S} \otimes 1_{\HA} \Psi^{(0)}
\end{equation}
At the first order of perturbation, the nonlinearity is replaced by the pre-integration of the zero-order solution. Suppose that $\Psi^{(0)}(t=0) = \psi(t=0) \otimes \xi_\alpha$, then $\Psi^{(0)}(t) = \psi^{(0)}(t) \otimes \xi_\alpha$ with $\ihbar \dot \psi^{(0)} = H_{\mathcal S} \psi^{(0)}$ and
\begin{equation}
W^{(0)} = |\psi^{(0)}\rangle \langle \xi_\alpha| \quad \text{and} \quad  W^{(0)\pinv} = |\xi_\alpha \rangle \langle \psi^{(0)}|
\end{equation}
\begin{equation}
W^{(0)\pinv} \Gamma_k W^{(0)} = \langle \psi^{(0)} |\Gamma_k|\psi^{(0)}\rangle |\xi_\alpha \rangle \langle \xi_\alpha|
\end{equation}
Let $\lambda_b^{(0)}(t)$ be the instantaneous eigenvalues of $H_{\mathcal S}(t)$ (supposed non-degenerate) and $\underline \zeta_b^{(0)}(t)$ be the associated normalized eigenvectors:
\begin{equation}
H_{\mathcal S} \underline \zeta_b^{(0)} = \lambda_b^{(0)} \underline \zeta_b^{(0)}
\end{equation}
At the first order of the perturbation, the eigenvalues and the eigenvectors of $H^{eff}(t)$ are
\begin{eqnarray}
\mu_b & = & \mu_b^{(0)} - \frac{\imath}{2} \gamma^k \langle \underline \zeta_b^{(0)} | \Gamma_k^\dagger \Gamma_k |\underline \zeta_b^{(0)}\rangle + \mathcal O(|\gamma|^2) \\
\underline \zeta_b & = & \underline \zeta_b^{(0)} - \frac{\imath}{2} \sum_{c \not= b} \gamma^k \frac{\langle \underline \zeta_c^{(0)}|\Gamma_k^\dagger \Gamma_k|\underline \zeta_b^{(0)} \rangle}{\mu_b^{(0)} - \mu_c^{(0)} - \frac{\imath}{2} \gamma^l (\langle \underline \zeta_b^{(0)} | \Gamma_l^\dagger \Gamma_l |\underline \zeta_b^{(0)}\rangle - \langle \underline \zeta_c^{(0)} | \Gamma_l^\dagger \Gamma_l |\underline \zeta_c^{(0)}\rangle)} \underline \zeta_c^{(0)} \nonumber \\
& & \qquad \qquad + \mathcal O(|\gamma|^2)
\end{eqnarray}
The eigenvalues and the eigenvectors of $H_{\mathcal U}$ are then
\begin{eqnarray}
\lambda_{b\beta}  & = & \mu_b + \frac{\imath}{2} \delta_{\alpha \beta} \gamma^k \langle \underline \zeta_b^{(0)}|\Gamma_k|\underline \zeta_b^{(0)}\rangle \langle \psi^{(0)}|\Gamma_k^\dagger | \psi^{(0)} \rangle + \mathcal O(|\gamma|^2) \\
\underline \Phi_{b\beta} & = & \underline \zeta_b \otimes \xi_\beta \nonumber \\
& & + \frac{\imath}{2} \delta_{\alpha \beta} \sum_{d\not=b} \gamma^k \frac{\langle \underline \zeta_d^{(0)} |\Gamma_k | \underline \zeta_b^{(0)} \rangle \langle \psi^{(0)}|\Gamma_k^\dagger |\psi^{(0)} \rangle}{\mu_b-\mu_d+ \frac{\imath}{2} \gamma^l (\langle \underline \zeta_b^{(0)}|\Gamma_l|\underline \zeta_b^{(0)}\rangle- \langle \underline \zeta_d^{(0)}|\Gamma_l|\underline \zeta_d^{(0)}\rangle) \langle \psi^{(0)}|\Gamma_l^\dagger | \psi^{(0)} \rangle } \underline \zeta_d^{(0)} \otimes \xi_\alpha \nonumber \\
\label{phibbeta} & & + \mathcal O(|\gamma|^2)
\end{eqnarray}
where $(\xi_\beta)_\beta$ is an arbitrary (time independent) basis of $\HA$. The eigenvectors of $H_{\mathcal U}^\dagger$ are
\begin{eqnarray}
  \underline \Phi_{b\beta}^*  & = &  \underline \zeta_b^* \otimes \xi_\beta \nonumber \\
  & & - \frac{\imath}{2} \delta_{\alpha \beta} \sum_{d\not=b} \gamma^k \frac{\langle \underline \zeta_d^{(0)} |\Gamma_k^\dagger | \underline \zeta_b^{(0)} \rangle \langle \psi^{(0)}|\Gamma_k |\psi^{(0)} \rangle}{\mu_b-\mu_d- \frac{\imath}{2} \gamma^l (\langle \underline \zeta_b^{(0)}|\Gamma_l^\dagger|\underline \zeta_b^{(0)}\rangle- \langle \underline \zeta_d^{(0)}|\Gamma_l^\dagger|\underline \zeta_d^{(0)}\rangle) \langle \psi^{(0)}|\Gamma_l | \psi^{(0)} \rangle } \underline \zeta_d^{(0)} \otimes \xi_\alpha \nonumber \\
  & & + \mathcal O(|\gamma|^2)
\end{eqnarray}
with $\underline \zeta_b^*$ the eigenvectors of $H^{eff\dagger}$ which are biorthonormal to the eigenvectors of $H^{eff}$ ($\langle \underline \zeta_c^*|\underline \zeta_b \rangle = \delta_{cb}$).\\

We suppose that $\Psi(t=0) = \underline \Phi_{a \alpha}(t=0)$, $\lambda_{a\alpha}$ being non-degenerate. If the evolution is slow and under some technical assumptions, we can prove the following adiabatic transport formula (see ref.\cite{Viennot4}) :
\begin{equation}
\rho(t) = g(t) \rho_{a\alpha}(t) g(t)^\dagger + \mathcal O\left(\max\left(\frac{1}{T},|\gamma|^2\right)\right)
\end{equation}
\begin{equation}
g(t) = \Teg^{-\ihbar^{-1} \int_0^T E_\alpha^{(1)}(t)dt} \Ted^{- \int_0^T A_\alpha^{(1)}(t) dt}
\end{equation}
with
\begin{equation}
\rho_{a\alpha} = \tr_{\HA} |\underline \Phi_{a\alpha} \rrangle \llangle \underline \Phi_{a\alpha}| = |\underline \zeta_{a\alpha}^{(1)} \rangle \langle \underline \zeta_{a\alpha}^{(1)} | + \mathcal O(|\gamma|^2)
\end{equation}
\begin{equation}
E_\alpha^{(1)} = \sum_b \lambda_{b\alpha} |\underline \zeta_{b\alpha}^{(1)} \rangle \langle \underline \zeta_{b\alpha}^{(1)*}|
\end{equation}
\begin{equation}
A_\alpha^{(1)} = \sum_{bc} \langle \underline \zeta_{b\alpha}^{(1)*} | \frac{d}{dt}| \underline \zeta_{c\alpha}^{(1)} \rangle |\underline \zeta_{b\alpha}^{(1)} \rangle \langle \underline \zeta_{c\alpha}^{(1)*} |
\end{equation}
\begin{equation}
\underline \zeta_{b\alpha}^{(1)} = \underline \zeta_b + \frac{\imath}{2} \sum_{d \not=b} \gamma^k \frac{\langle \underline \zeta_d^{(0)}|\Gamma_k|\underline \zeta_b^{(0)} \rangle \langle \psi^{(0)}|\Gamma_k^\dagger|\psi^{(0)} \rangle}{\mu_b - \mu_d + \frac{\imath}{2} \gamma^l (\langle \underline \zeta_b^{(0)}|\Gamma_l|\underline \zeta_b^{(0)} \rangle - \langle \underline \zeta_d^{(0)}|\Gamma_l|\underline \zeta_d^{(0)} \rangle)\langle \psi^{(0)}|\Gamma_l^\dagger|\psi^{(0)} \rangle} \underline \zeta_d^{(0)}
\end{equation}
\begin{equation}
\underline \zeta_{b\alpha}^{(1)*} = \underline \zeta_b - \frac{\imath}{2} \sum_{d \not=b} \gamma^k \frac{\langle \underline \zeta_d^{(0)}|\Gamma_k^\dagger|\underline \zeta_b^{(0)} \rangle \langle \psi^{(0)}|\Gamma_k|\psi^{(0)} \rangle}{\mu_b - \mu_d - \frac{\imath}{2} \gamma^l (\langle \underline \zeta_b^{(0)}|\Gamma_l^\dagger|\underline \zeta_b^{(0)} \rangle - \langle \underline \zeta_d^{(0)}|\Gamma_l^\dagger|\underline \zeta_d^{(0)} \rangle)\langle \psi^{(0)}|\Gamma_l|\psi^{(0)} \rangle} \underline \zeta_d^{(0)}
\end{equation}

The adiabatic geometric phase issued from the adiabatic theorem\cite{Viennot4} $\Ted^{- \int_0^T A_\alpha^{(1)}(t) dt}$ is related to the non-adiabatic geometric phase as follows. The role of $\tilde \rho(t)$ in theorem \ref{thgeoph} is played by $\rho_{a\alpha}(t)$. We have $W_{\rho_{a\alpha}} = |\underline \zeta_{a\alpha}^{(1)}\rangle \langle \xi_\alpha|$ and $W_{\rho_{a\alpha}}^\pinv = |\xi_\alpha \rangle \langle \underline \zeta_{a\alpha}^{(1)*}|$. It follows that
\begin{equation}
i_{\dot W_{\rho_{a\alpha}}} \mathfrak A = \dot W_{\rho_{a\alpha}} W_{\rho_{a\alpha}}^\pinv = |\frac{d}{dt} \underline \zeta_{a\alpha}^{(1)} \rangle \langle \underline \zeta_{a\alpha}^{(1)*}|
\end{equation}
and
\begin{equation}
A^{(1)}_\alpha = i_{\dot W_{\rho_{a\alpha}}} \mathfrak A + \sum_{c\not=a} |\frac{d}{dt} \underline \zeta_{c\alpha}^{(1)} \rangle \langle \underline \zeta_{c\alpha}^{(1)*}| + \mathcal O(|\gamma|^2)
\end{equation}
where we have use the fact that $(\underline \zeta_{b\alpha}^{(1)};\underline \zeta_{b\alpha}^{(1)*})_b$ constitutes (at the order $|\gamma|^2$) a biorthonormalized basis of $\HS$ ($\sum_b |\underline \zeta_{b\alpha}^{(1)} \rangle \langle\underline \zeta_{b\alpha}^{(1)*}| =1$ ). Moreover
\begin{equation}
  \sum_{c\not=a} |\frac{d}{dt} \underline \zeta_{c\alpha}^{(1)} \rangle \langle \underline \zeta_{c\alpha}^{(1)*}| \rho_{a\alpha} + \rho_{a\alpha} \left(\sum_{c\not=a} |\frac{d}{dt} \underline \zeta_{c\alpha}^{(1)} \rangle \langle \underline \zeta_{c\alpha}^{(1)*}| \right)^\dagger = \mathcal O(|\gamma|^2)
\end{equation}
It follows that $\eta = \sum_{c\not=a} |\frac{d}{dt} \underline \zeta_{c\alpha}^{(1)} \rangle \langle \underline \zeta_{c\alpha}^{(1)*}| \in \mathfrak{gl}(\HS)_{\rho_{a\alpha}}$ and then
\begin{equation}
A^{(1)}_\alpha = i_{\dot W_{\rho_{a\alpha}}} \mathfrak A + \eta
\end{equation}
The adiabatic geometric phase generator $A^{(1)}_\alpha$ found in ref.\cite{Viennot4} is the sum of the two geometric phase generators found theorem \ref{thgeoph} with the ``density eigenmatrix'' $\rho_{a\alpha}$ playing the role of $\tilde \rho(t)$.\\
Let $P_{\bullet \alpha} = \sum_b |\underline \Phi_{b\alpha} \rrangle \llangle \underline \Phi_{b\alpha}^*|$ be the eigenprojection in $\HS \otimes \HA$ associated with the eigenvectors of $H_{\mathcal U}$ related to $\xi_\alpha$ by perturbation. We consider the following reduced generator of geometric phases:
\begin{eqnarray}
\tr_{\HA} \left(|P_{\bullet \alpha} \frac{d}{dt} \underline \Phi_{a\alpha} \rrangle \llangle \underline \Phi_{a\alpha}^*| \right) \rho_{a\alpha}^\pinv & = & \sum_b \langle \underline \zeta_{b\alpha}^{(1)*}|\frac{d}{dt}|\underline \zeta_{a\alpha}^{(1)} \rangle | \underline \zeta_{b\alpha}^{(1)} \rangle \langle  \underline \zeta_{a\alpha}^{(1)*}| + \mathcal O(|\gamma|^2) \\
& = & |\frac{d}{dt} \underline \zeta_{a\alpha}^{(1)} \rangle \langle \underline \zeta_{a\alpha}^{(1)*}| + \mathcal O(|\gamma|^2)
\end{eqnarray}
\begin{eqnarray}
\PSI^{-1}(P_{\bullet \alpha})[\dot W_{\rho_{a\alpha}}] W_{\rho_{a\alpha}}^\pinv & = & A_\alpha^{(1)} + \mathcal O(|\gamma|^2) \\
& = & i_{\dot W_{\rho_{a\alpha}}} \mathfrak A + \eta + \mathcal O(|\gamma|^2) 
\end{eqnarray}
Moreover
\begin{eqnarray}
E_\alpha^{(1)} \otimes |\xi_\alpha \rangle \langle \xi_\alpha| & = & P_{\bullet \alpha} H_{\mathcal U}(\underline \Phi_{a\alpha}) P_{\bullet \alpha} \\
& = & E(\rho_{a\alpha}) \otimes |\xi_\alpha \rangle \langle \xi_\alpha| + \mathcal O(|\gamma|^2)
\end{eqnarray}
The geometric and dynamical phase generators appearing in the adiabatic theorem ref.\cite{Viennot4} are clearly the adiabatic limit of the generators introduced in theorem \ref{thgeoph}.\\
Let $G_{\{\underline \Phi_{b\beta}\}_{b\beta}}(t) \subset U(\HS)$ be the group of operators of $\HS$ acting on the eigenvectors of $H_{\mathcal U}$ as phase changes. By definition we have
\begin{eqnarray}
\forall g \in G_{\{\underline \Phi_{b\beta}\}_{b\beta}}, \quad H_{\mathcal U}(g\underline \Phi_{a\alpha}) g\underline \Phi_{a\alpha} & = & \lambda_{a\alpha} g \underline \Phi_{a\alpha} + \mathcal O(|\gamma|^2) \\
& = & g H_{\mathcal U}(\underline \Phi_{a\alpha}) \underline \Phi_{a\alpha} + \mathcal O(|\gamma|^2)
\end{eqnarray}
Clearly $E(g \rho_{a\alpha} g^{-1}) = E(\rho_{a\alpha})$ and $G_{\{\underline \Phi_{b\beta}\}_{b\beta}} \subset G_{\mathcal L}$. The adiabatic geometric phase is a phase with respect to the generator of the dynamics and $A_\alpha^{(1)} = i_{\dot W_{\rho_{a\alpha}}} A^L$. Finally the adiabatic assumption can be rewritten as $i_{\dot W_{\rho_{a\alpha}}} \mathfrak A \simeq  i_{\dot W_{\rho_{a\alpha}}} A^L$ in accordance with the discussion found in ref.\cite{Viennot1}. The geometric structure involved by the adiabatic geometric phases has been extensively studied in ref.\cite{Viennot1} and the interpretation of the adiabatic operator valued geometric phase has been studied in ref. \cite{Viennot3}.\\
This approach for the adiabatic geometric phases, which is usefull for bipartite quantum systems, is strongly limited for the open quantum systems described by a Lindblad equation. Indeed both the perturbative assumption and the adiabatic assumption are too drastic because they involve that the dynamics remains in the neighbourhood of the singular stratum of the pure states ($\rho_{a\alpha}$ is a pure state at the first order of perturbation). This excludes the more interesting cases where the relaxation rates are sufficiently large to deviate the dynamics or where the interaction duration is sufficiently large to the dynamics goes to the steady state (with the open quantum systems there is a competition between the adiabatic and the relaxation processes). In second order of perturbation the adiabatic approximation deals with mixed state \cite{Viennot4} but the nonlinearity could induces strong difficulties to apply this approach. A more general approach of the adiabatic geometric phases could consists to use the notion of noncommutative eigenvalues as introduced in ref. \cite{Viennot1} for bipartite quantum systems, but the nonlinearity could also induces strong difficulties.\\

Table \ref{geomph} summarizes the different geometric phases of the open quantum systems.

\begin{table}
\caption{\label{geomph} The different geometric phases of the open quantum systems.}
\rotatebox{90}{\begin{tabular}{|r||c|c|c|c|l|}
\hline
\textit{geometric phase} & \textit{generator} & \textit{values} & \textit{side} & \textit{type} & \textit{interpretation}\\
\hline \hline
$C^*$-geometric phase & $\mathfrak A$ & $G$ & left & by equivariance & transformations to cyclic density operators \\
\hline
Uhlmann geometric phase & $A^R$ & $K$ & right & by invariance & transition probability parallelism \\
\hline
Sj\"oqvist geometric phase & $A^{BL}/A^{BR}$ & $\overline{H^i}$ & two-sided & by invariance  & interferometric phase \\
\hline
reduced $C^*$-geometric phase & $A^L$ & $G_{\mathcal L}$ & left & w.r.t. the dynamics & closest to the adiabatic phase \\
\hline
\end{tabular}}
\end{table}

\section{Geometry of mixed states}
\subsection{The geometry as a stratified principal composite bibundle}
\subsubsection{Regular strata}
Let $\Sigma^i(\HS)$ be a regular stratum, $\mathcal D^i(\HS) = \pi^{-1}_D(\Sigma^i(\HS))$ be the stratum of the density operators over $\Sigma^i(\HS)$, and $\mathcal B^i(\HS) = \pi^{-1}_{HS}(\HS)$ be the stratum of the standard purified states over $\mathcal D^i(\HS)$.\\
Let $L_{HS}: G \to \Aut \mathcal B^i(\HS)$ be the left action of $G$ on $\mathcal B^i(\HS)$ defined by
\begin{equation}
\forall g\in G, \forall W \in \mathcal B^i(\HS) \qquad L_{HS}(g) W = gW
\end{equation}
and $R_{HS}: K \to \Aut \mathcal B^i(\HS)$ be the right action of $K$ on $\mathcal B^i(\HS)$ defined by
\begin{equation}
\forall k\in K, \forall W \in \mathcal B^i(\HS) \qquad R_{HS}(k) W = Wk
\end{equation}
These actions of $G$ and $K$ are free. The right action of $K$ can be locally transformed as a left action of $\overline{H^i}$, indeed
\begin{equation}
R_{HS}(k) W = L_{HS}(WkW^{-1}) W
\end{equation}
Because $\pi_{HS}(R_{HS}(k)W) = \pi_{HS} (L_{HS}(WkW^{-1})W) = \pi_{HS}(W)$ ($(WkW^{-1})WW^\dagger (WkW^{-1})^\dagger = WW^\dagger$), we have $WkW^{-1} \in GL(\HS)_{WW^\dagger} \subset \overline{H^i}$.\\
$\mathcal B^i(\HS)$ cannot be viewed as a left-bundle on $\mathcal D^i(\HS)$ since $\pi_{HS}(gW) = g\pi_{HS}(W)g^\dagger$, but since $\pi_D \circ \pi_{HS} (gW) = \pi_D \circ \pi_{HS} (W)$ we can see $\mathcal B^i(\HS)$ as a left principal $G$-bundle $\mathsf P^i$ over $\Sigma^i(\HS)$. Let $\phi_{\mathsf P}^i$ be the trivialization of $\mathsf P^i$:
\begin{equation}
\begin{array}{rcl} \phi_{\mathsf P}^i: G \times \Sigma^i(\HS) & \to & \mathcal B^i(\HS) \\ (g,\sigma) & \mapsto & g \sqrt \sigma \end{array}
\end{equation}
$\mathsf P^i$ is a trivial bundle because of the presence of the global section $\sigma \mapsto \sqrt \sigma$.\\
Since $\pi_{HS}(Wk) = \pi_{HS}(W)$, we can see $\mathcal B^i(\HS)$ as a righ principal $K$-bundle $\mathsf Q^i$ over $\mathcal D^i(\HS)$. Let $\phi^i_{\mathsf Q}$ be the trivialization of $\mathsf Q^i$:
\begin{equation}
\begin{array}{rcl} \phi_{\mathsf Q}^i: \mathcal D^i(\HS) \times K & \to & \mathcal B^i(\HS) \\ (\rho,k) & \mapsto & \sqrt \rho k \end{array}
\end{equation}
$\mathsf Q^i$ is a trivial bundle because of the presence of the global section $\rho \mapsto \sqrt \rho$.\\
Let $\pi_G^i : G \to G/H^i$ be the canonical projection ($\pi_G^i(g) = gH^i$). Let $\phi^i_D$ be the trivialization of $\mathcal D^i(\HS)$ viewed as a bundle of manifolds:
\begin{equation}
\begin{array}{rcl} \phi_{D}^i: G/H^i \times \Sigma^i(\HS) & \to & \mathcal D^i(\HS) \\ (gH^i,\sigma) & \mapsto & g\sigma g^\dagger \end{array}
\end{equation}
Let $\pi_{\bar H}^i : \overline{H^i} \to G/H^i$ be the projection defined by $\pi^i_{\bar H}(h) = \pi^i_G(h)$ ($\overline {H^i} \subset G$) which is such that $\pi_{\bar H}^i(g H^i g^{-1}) = gH^i$.\\
The complete geometric structure is then defined by the following commutative diagram:
$$ \begin{CD}
\overline{H^i} \times \Sigma^i(\HS) @>{\hookrightarrow}>> G \times \Sigma^i(\HS) @>{\simeq}>{\phi^i_{\mathsf P}}> \mathcal B^i(\HS) @<{\simeq}<{\phi^i_{\mathsf Q}}< \mathcal D^i(\HS) \times K \\
@V{\pi_{\bar H}^i \times \id}VV  @V{\pi^i_G \times \id}VV @V{\pi_{HS}}VV @V{\Pr_1}VV \\
G/H^i \times \Sigma^i(\HS) @= G/H^i \times \Sigma^i(\HS) @>{\simeq}>{\phi^i_D}> \mathcal D^i(\HS) @= \mathcal D^i(\HS) \\
& & @V{\Pr_2}VV @V{\pi_D}VV \\
& & \Sigma^i(\HS) @= \Sigma^i(\HS)
\end{CD} $$
Due to the presence of three floors ($\mathcal B^i(\HS)$, $\mathcal D^i(\HS)$ and $\Sigma^i(\HS)$) the bundle of purification is a composite bundle \cite{Sardanashvily, Viennot5, Viennot6}. Due to the both left and right structures, the bundle of purification is a bibundle \cite{Aschieri}. We can then say that $\mathcal B^i(\HS)$ is a principal composite bibundle.\\

The bundle can be rewritten with an arbitrary purification space $\HS \otimes \HA$ with the left and right actions defined by
\begin{eqnarray}
\forall g \in G, \forall \Psi \in \HS \otimes \HA, \quad L_{\Psi}(g) \Psi & = & g \otimes 1_\HA \Psi \\
\forall k \in K, \forall \Psi \in \HS \otimes \HA, \quad R_{\Psi}(k) \Psi & = & 1_\HS \otimes k^t \Psi
\end{eqnarray}
with $L_{HS}(g) \PSI(W) = \PSI(L_{\Psi}(g) W)$ and $R_{HS}(k) \PSI(W) = \PSI(R_{\Psi}(k) \PSI(W))$.

\subsubsection{Singular strata}
Let $\Sigma^i(\HS)$ be a singular stratum charaterized by $p_1=0$ with a degeneracy equal to $n-m$. A density operator $\rho \in \mathcal D^i(\HS)$ can be defined by two kinds of data: $P_{\Ran \rho}$ the projection onto its support space and an order $m$ invertible density matrix $\varrho \in \mathfrak M_{m\times m}(\mathbb C)$ representing $\rho$ in this space. Let $(\chi_j)_{j=1,...,m}$ be an orthonormal basis of $\Ran \rho$. We have:
\begin{equation}
\rho = {\varrho^j}_k |\chi_j \rangle \langle \chi^k| \qquad P_{\Ran \rho} = |\chi_j \rangle \langle \chi^j|
\end{equation}
with ${\varrho^j}_k = \langle \chi^j|\rho \chi_k \rangle$. Reciprocally we can introduce $\RHO : \mathcal D^i(\mathbb C^m) \times V_m(\HS) \to \mathcal D^i(\HS)$ defined by
\begin{eqnarray}
\forall \varrho \in \mathcal D^i(\mathbb C^m), \forall (\chi_j)_{j=1,...,m} \in V_m(\HS), \quad \RHO(\varrho,(\chi_j)_j) = {\varrho^j}_k |\chi_j\rangle \langle \chi^k|
\end{eqnarray}
where $\mathcal D^i(\mathbb C^m)$ is the regular strata for the model Hilbert space $\mathbb C^m$ such that $\Sigma^i(\mathbb C^m)=\pi_D(\mathcal D^i(\mathbb C^m))$ is similar to $\Sigma^i(\HS)$ deprived of the zeros. $V_m(\HS) = \{(\chi_j)_{j=1,...,m} \in (\HS)^m, \langle \chi^j|\chi_k \rangle = \delta^j_k \} \simeq V_m(\mathbb C^n) = U(n)/U(n-m)$ ($V_m(\mathbb C^n)$ is a Stiefel manifold).\\
Let $I^a \subset \{1,...,n\}$ be a set of indices such that $\dist_{FS}(P^a,P_{\Ran \rho}) < \frac{\pi}{2}$ where $P^a = \sum_{j \in I^a} |\zeta_j \rangle \langle \zeta_j|$ and $\dist_{FS}(P,Q) = \arccos|\det(Z_P^\dagger Z_Q)|^2$ is the Fubini-Study distance on $G_m(\HS) = \{P \in \mathcal B(\HS), P^\dagger = P, P^\dagger=P, \dim \Ran P = m\} \simeq G_m(\mathbb C^n)$ ($Z_P$ is the matrix representing an arbitrary orthonormal basis of $\Ran P$ in the basis $(\zeta_j)_j$). There exists a basis $(\chi_j)_j$ of $\Ran \rho$ such that\cite{Rohlin}
\begin{equation}
\exists c_{jk} \in \mathbb C, \quad \chi_j = \zeta_j + \sum_{k \not\in I^a} c_{jk} \zeta_k
\end{equation}
$U^a = \{P \in G_m(\HS), \dist_{FS}(P,P^a) < \frac{\pi}{2}\}$ is an open chart of $G_m(\HS)$ with the coordinates map $\begin{array}{rcl} \xi^a : U^a & \to & \mathbb C^{m(n-m)} \\ P & \mapsto & (c_{jk})_{j\in I^a, k\not\in I^a} \end{array}$. $V_m(\HS)$ is a non-trivial principal $U(m)$-bundle over $G_m(\HS)$ with local trivializations
\begin{equation}
\begin{array}{rcl} \phi^a_S: U^a \times U(m) & \to & V_m(\HS)_{|U^a} \\ (P,u) & \mapsto & {(Z^a_Pu)^j}_k |\zeta_j \rangle \langle \chi^k| \end{array}
\end{equation}
with $Z^a_P \in \mathfrak M_{m \times m}(\mathbb C)$ be such that ${(Z^a_P)^j}_k = \langle \zeta^j|\chi_k \rangle$ ($(\chi_j)_{j \in I^a}$ is an orthonormal basis of $\Ran P$).\\
Let $H^{i0}$ and $H^{i1}$ be such that $H^i = H^{i0} \times H^{i1}$ with $H^{i0} \simeq GL(m,\mathbb C)_\sigma$ ($\sigma \in\Sigma^i(\mathbb C^m)$) and $H^{1i} \simeq GL(n-m,\mathbb C)$. $G/H^{i1} \simeq \tilde V_m(\mathbb C^n) = GL(n,\mathbb C)/GL(n-m,\mathbb C) = V_m(\mathbb C^n) \times_{U(m)} GL(m,\mathbb C) = \{[(\chi u^{-1}, ug), u \in U(m)],\chi \in V_m(\mathbb C^n), g\in GL(m,\mathbb C)\}$ is a non-compact Stiefel manifold. Let $\hat \phi^{ia}_S$ be the local trivialization of $G/H^{i1}$ induced by\\
$\begin{array}{rcl} \tilde \phi^{ia}_S : U^a \times GL(m,\mathbb C) & \to & \tilde V_m(\mathbb C^n) \\ (P,g) & \mapsto & [(Z^a_P u^{-1},ug), u\in U(m)] \end{array}$.\\
The complete geometric structure is defined by the following commutative diagram:\\
\rotatebox{90}{\begin{minipage}{19cm} \footnotesize
$$ \begin{CD}
\overline{H^i} \times \Sigma^i(\HS) @>{\hookrightarrow}>> G \times \Sigma^i(\HS) \\
& &  @V{\pi^{i1}_G \times \id}VV \\
G/H^i \times H^{i0} \times \Sigma^i(\HS) @>{\approx}>{\{\hat \phi^{ia}_S\}_a \times \id}> G/H^{i1} \times \Sigma^i(\HS) @>{\simeq}>{\phi^i_{\mathsf P}}> \mathcal B^i(\HS) @<{\simeq}<{\RHO^*}< \mathcal D^i(\mathbb C^m) \times V_m(\HS) @<{\approx}<{\id \times \{\phi^a_S\}_a}< \mathcal D^i(\mathbb C^m) \times G_m(\HS) \times U(m) \\
@V{\Pr_1 \times \Pr_3}VV @V{\pi^{i0}_G \times \id}VV @V{\pi_{HS}}VV @V{\id \times \pi_S}VV @V{\Pr_1 \times \Pr_2}VV \\
G/H^i \times \Sigma^i(\HS) @= G/H^i \times \Sigma^i(\HS) @>{\simeq}>{\phi^i_D}> \mathcal D^i(\HS) @<{\simeq}<{\RHO_*}< \mathcal D^i(\mathbb C^m) \times G_m(\HS) @= \mathcal D^i(\mathbb C^m) \times G_m(\HS) \\
& & @V{\Pr_2}VV @V{\pi_D}VV @V{\pi_D \circ \Pr_1}VV \\
& & \Sigma^i(\HS) @= \Sigma^i(\HS) @<{\simeq}<< \Sigma^i(\mathbb C^m)
  \end{CD} $$
\end{minipage}}\\
where $\pi^{i1}_G : G \to G/H^{i1}$ and $\pi^{i0}_G : G/H^{i1} \to G/(H^{i1} \times H^{i0})$ are the canonical projections; $\RHO_*(\varrho,P) = \RHO(\varrho,(\chi_i)_i)$ for an arbitrary orthonormal basis $(\chi_i)_i$ of $\Ran P$ and 
\begin{equation}
\begin{array}{rcl} \RHO^* : \mathcal D^i(\HS) \times V_m(\HS) & \to & \mathcal B^i(\HS) \\ (\varrho ,\chi) & \mapsto & {\sqrt \varrho^i}_j |\chi_i \rangle \langle \chi^j | \end{array}
\end{equation}
$\mathcal B^i(\HS)$ can be viewed as a (non-trivial) right principal $U(m)$-bundle $\mathsf Q^i$ over $\mathcal D^i(\HS)$ with local trivializations $\RHO^* \circ \phi^a_S \circ \RHO_*^{-1}$. Note that $U(m)$ acts on $\mathcal B^i(\HS)$ by the right action:
\begin{equation}
\forall u \in U(m), \forall W=\RHO^*(\varrho_W,\phi^a(g_W,u_W)) \in \mathcal B^i(\HS), \quad R_{loc}(u)W = \RHO^*(\varrho_W,\phi^a_S(g_W,u_W u))
\end{equation}
By construction $\forall u \in U(m)$, $\exists k_{u,W} \in K$ such that $R_{loc}(u) = R_{HS}(k_{u,W})$. We denote by $K_W^i$ the subgroup of $K$ such that $R_{loc}(U(m))_{|W} = R_{HS}(K_W^i)$. On the left, $\mathcal B^i(\HS)$ is a trivial bundle $\mathsf P^i$ over $\Sigma^i(\HS)$ with typical fiber $G/H^{i1}$ (which is not a group since $H^{i1}$ is not normal).

\subsection{The geometry as categorical principal bundles}
We want to endow the geometric structure with a connective structure describing the geometric phases. On the regular strata, the geometric structure is a principal composite bibundle. A connective structure on a principal bibundle \cite{Aschieri} and a connective structure on a principal composite bundle \cite{Viennot6} present a higher degree of complication than for an usual principal bundle, the compatibility between the two approaches can be also a difficulty. In order to avoid this, and to unify the geometric description between the left and the right, we propose to use a categorical geometric generalization of the principal bundle\cite{Baez1,Baez2,Baez3,Wockel,Nikolaus,Viennot7}.
\subsubsection{Left principal categorical bundle}
We consider a stratum $\Sigma^i(\HS)$. Let $\mathscr S^i$ be the (trivial) category with the set of objects $\Obj \mathscr S^i = \Sigma^i(\HS)$ and the set of morphisms $\Morph \mathscr S^i = \{\id_\sigma, \sigma \in \Sigma^i(\HS) \}$. Let $\mathscr P^i$ be the category with $\Obj \mathscr P^i = \mathcal B^i(\HS)$ and $\Morph \mathscr P^i = \{(W,k), W \in \mathcal B^i(\HS), k \in K^i_W\}$ ($K^i_W = K$ on the regular strata) with the source, the target and the identity maps defined by
\begin{eqnarray}
\forall (W,k) \in \Morph \mathscr P^i, \quad s(W,k) & = & W \\
t(W,k) & = & R_{HS}(k)W = Wk\\
\id_W & = & (W,1_\HS)
\end{eqnarray}
the arrow composition being
\begin{equation}
\forall W \in \mathcal B^i(\HS), \forall k \in K^i_W, \forall k'\in K^i_{Wk}, \quad (Wk,k') \circ (W,k) = (W,kk')
\end{equation}
By construction, arrows of different strata are not composable. Let $\varpi_D \in \Funct(\mathscr P^i,\mathscr S^i)$ be the full functor defined by
\begin{eqnarray}
\forall W \in \Obj \mathscr P^i, \quad \varpi_D(W) & = & \pi_D \circ \pi_{HS}(W) \\
\forall (W,k) \in \Morph \mathscr P^i, \quad \varpi_D(W,k) & = & \id_{\pi_D \circ \pi_{HS}(W)}
\end{eqnarray}
Note that $\varpi_D(t(W,k)) = \varpi_D(Wk) = \pi_D(Wkk^{-1} W^\dagger) = \pi_D(WW^\dagger) = t(\varpi_D(W,k))$.\\
We can remark that the arrows could also be defined by an action on the left: $\Morph \mathscr P^i = \{(h,W), W\in \mathcal B^i(\HS), h \in WK^i_WW^\pinv \subset \overline{H^i} \}$ with $t(h,W) = L_{HS}(h)W$, since $\exists k \in K$ such that $h=WkW^\pinv$ and $L_{HS}(WkW^\pinv)W = WkWW^\pinv = Wk(1-P_{\ker W}) = Wk$ ($kP_{\ker W}=0$, $\forall k \in K_W^i$).\\
Let $\mathscr G^i$ be the groupoid defined by $\Obj \mathscr G^i = G$ and $\Morph \mathscr G^i = G \rtimes \overline{H^i}$ with the source, the target and the identity maps defined by
\begin{eqnarray}
\forall (g,h) \in \Morph \mathscr G^i, \quad s(g,h) & = & g \\
t(g,h) & = & gh \\
\id_g & = & (g,1_\HS)
\end{eqnarray}
the usual arrow composition (called the vertical composition of arrows) being
\begin{equation}
\forall g\in G, \forall h,h'\in \overline{H^i}, \quad (gh,h') \circ (g,h) = (g,hh')
\end{equation}
and the law of the semi-direct product of groups (called the horizontal composition of arrows) being
\begin{equation}
\forall (g,h),(g',h')\in G\rtimes \overline{H^i}, \quad (g',h')(g,h) = (g'g,g^{-1}h'gh)
\end{equation}
Let $\mathscr L : \Morph \mathscr G^i \to \EndFunct(\mathscr P^i)$ be the left action of $\mathscr G^i$ onto $\mathscr P^i$ defined by
\begin{eqnarray}
\forall (g,h) \in \Morph \mathscr G^i, \forall W \in \Obj \mathscr P^i, \quad \mathscr L_{g,h}(W) & = & L_{HS}(gh)W \\
\forall (W,k) \in \Morph \mathscr P^i, \quad \mathscr L_{g,h}(W,k) & = & (ghW,k)
\end{eqnarray}
or similarly by
\begin{equation}
\begin{CD} W \\ @V{(W,k)}VV \\ Wk \end{CD} \overset{\mathscr L_{g,h}}{\Verylongrightarrow} \begin{CD} ghW \\ @VV{(ghW,k)}V \\ ghWk \end{CD}
\end{equation}
The composition of left actions is covariant with the horizontal composition of the groupoid arrows:
\begin{equation}
\mathscr L_{g',h'} \circ \mathscr L_{g,h} = \mathscr L_{g'g,g^{-1}h'gh} = \mathscr L_{(g',h')(g,h)}
\end{equation}
The compatibility between the projection functor $\varpi_D$ and the left action endofunctor $\mathscr L$ is shown by the following commutative diagram (where $\sigma_W = \pi_D \circ \pi_{HS}(W)$):
\begin{equation}
\begin{array}{ccccc}
 & \rotatebox[origin=r]{45}{$\overset{\mathscr L_{g,h}}{\Verylongrightarrow}$} & \begin{CD} ghW \\ @VV{(ghW,k)}V \\ ghWk \end{CD} & \rotatebox{-45}{$\overset{\varpi_D}{\Verylongrightarrow}$} & \\
\begin{CD} W \\ @V{(W,k)}VV \\ Wk \end{CD} & & & & \begin{CD} \sigma_W \\ @VV{\id_{\sigma_W}}V \\ \sigma_W \end{CD} \\
& \rotatebox[origin=r]{-45}{$\overset{\varpi_D}{\Verylongrightarrow}$} & \begin{CD} \sigma_W \\ @VV{\id_{\sigma_W}}V \\ \sigma_W \end{CD} & \rotatebox{45}{$\Verylongeq$}
\end{array}
\end{equation}
$\mathscr P^i$ is naturally equivalent to $\mathscr G^i \times \mathscr S^i$ by the trivialization functors $\phi_{\mathscr P}^i \in \Funct(\mathscr G^i \times \mathscr S^i, \mathscr P^i)$ and $\bar \phi_{\mathscr P}^i \in \Funct(\mathscr P^i, \mathscr G^i \times \mathscr S^i)$ defined by
\begin{equation}
\begin{array}{ccc}
\begin{CD} (g,\sigma) \\ @V{(g,h,\sigma,\id_\sigma)}VV \\ (gh,\sigma) \end{CD} & \overset{\phi_{\mathscr P}^i}{\Verylongrightarrow} & \begin{CD} g\sqrt{\sigma} \\ @VV{(g\sqrt \sigma,k)}V \\ g\sqrt{\sigma}k \end{CD} \\
\\
\begin{CD} (W \sqrt{\sigma_W}^\pinv,\sigma_W) \\ @V{(W\sqrt{\sigma_W}^\pinv,h,\sigma_W,\id_{\sigma_W})}VV \\ Wk\sqrt{\sigma_W}^\pinv \end{CD} & \overset{\bar \phi_{\mathscr P}^i}{\Verylongleftarrow} & \begin{CD} W \\ @VV{(W,k)}V \\ Wk \end{CD}
\end{array}
\end{equation}
with $h = \sqrt \sigma k \sqrt \sigma^\pinv$ and $\sigma_W = \pi_D \circ \pi_{HS}(W)$. We can note that $\phi^i_{\mathscr P}(g,\sigma)=\phi^i_{\mathsf P}(g,\sigma)$.\\

Viewed as a left principal categorial bundle, the presence of the both left and right actions takes a natural meaning. The right action of $K$ (inobservable reconfiguration of the ancilla) defines the arrows of the purification category whereas the left action of $G$ (unitary and SLOCC transformations of the system) defines endofunctors of the purification category playing the role of gauge changes.

\subsubsection{Right principal categorical bundle}
It is possible to reverse the description, by considering that the right action of $K$ defines endofunctors of gauge changes (in accordance with their inobservable character) and that the left action of $G$ defines the arrows.\\

Let $\mathscr M^i$ be the category defined by $\Obj \mathscr M^i = \mathcal D^i(\HS)$ and $\Morph \mathscr M^i = G \times \mathcal D^i(\HS)$, with the source, the target, and the identity maps defined by
\begin{eqnarray}
\forall (g,\rho) \in \Morph \mathscr M^i, \quad s(g,\rho) & = & \rho \\
t(g,\rho) & = & g\rho g^\dagger \\
\id_\rho & = & (1_\HS, \rho)
\end{eqnarray}
the arrow composition being defined by
\begin{equation}
\forall \rho \in \mathcal D^i(\HS), \forall g,g'\in G, \quad (g',g\rho g^\dagger) \circ (g,\rho) = (g'g,\rho)
\end{equation}
Let $\mathscr Q^i$ be the category defined by $\Obj \mathscr Q^i = \mathcal B^i(\HS)$ and $\Morph \mathscr Q^i = G \times \mathcal B^i(\HS)$, with the source, the target and the identity maps defined by
\begin{eqnarray}
\forall (g,W) \in \Morph \mathscr Q^i, \quad s(g,W) & = & W \\
t(g,W) & = & gW \\
\id_W & = & (1_\HS,W)
\end{eqnarray}
the arrow composition being defined by
\begin{equation}
\forall W \in \mathcal B^i(\HS), \forall g,g'\in G, \quad (g',gW) \circ (g,W) = (g'g,W)
\end{equation}
Let $\varpi_{HS} \in \Funct(\mathscr Q^i,\mathscr M^i)$ be the full functor defined by
\begin{eqnarray}
\forall W \in \Obj \mathscr Q^i, \quad \varpi_{HS}(W) & = & \pi_{HS}(W)= WW^\dagger \\
\forall (g,W) \in \Morph \mathscr Q^i, \quad \varpi_{HS}(g,W) & = & (g,WW^\dagger)
\end{eqnarray}

Let $\mathscr K^i$ be the groupoid with $\Obj \mathscr K^i = K^i$ and $\Morph \mathscr K^i = K^i \ltimes K^i$ (where $K^i_W \simeq K^i$, $\forall W \in \mathcal B^i(\HS)$; $K^i = K$ on the regular strata), with the source, the target and the identity maps defined by
\begin{eqnarray}
\forall (q,k) \in \Morph \mathscr K^i, \quad s(q,k) & = & k \\
t(q,k) & = & qk \\
\id_k & = & (1_\HS,k)
\end{eqnarray}
the usual arrow composition (called the vertical composition of arrows) being
\begin{equation}
\forall k\in K^i, \forall q,q'\in K^i, \quad (q',qk) \circ (q,k) = (q'q,k)
\end{equation}
and the law of the semi-direct product of groups (called the horizontal composition of arrows) being
\begin{equation}
\forall (q,k),(q',k')\in K^i \ltimes K^i, \quad (q',k')(q,k) = (q'k'q(k')^{-1},k'k)
\end{equation}
Let $\mathscr R:\Morph \mathscr K^i \to \EndFunct(\mathscr Q^i)$ be the right action of $\mathscr K^i$ onto $\mathscr Q^i$ defined by
\begin{equation}
\begin{CD} W \\ @V{(g,W)}VV \\ gW \end{CD} \overset{\mathscr R_{q,k}}{\Verylongrightarrow} \begin{CD} Wqk \\ @VV{(g,Wqk)}V \\ gWqk \end{CD}
\end{equation}
where to simplify the notations, for the singular strata, we write $Wk$ in place of $R_{loc}(k)W$. The composition of right actions is contravariant with the horizontal composition of the groupoid arrows:
\begin{equation}
\mathscr R_{q,k} \circ \mathscr R_{q',k'} = \mathscr R_{q'k'q(k')^{-1},k'k} = \mathscr R_{(q',k')(q,k)}
\end{equation}
The compatibility between the projection functor $\varpi_{HS}$ and the right action endofunctor $\mathscr R$ is shown by the following commutative diagram (where $\rho_W = WW^\dagger$):
\begin{equation}
\begin{array}{ccccc}
 & \rotatebox[origin=r]{45}{$\overset{\mathscr R_{q,k}}{\Verylongrightarrow}$} & \begin{CD} Wqk \\ @VV{(g,Wqk)}V \\ gWqk \end{CD} & \rotatebox{-45}{$\overset{\varpi_{HS}}{\Verylongrightarrow}$} & \\
\begin{CD} W \\ @V{(g,W)}VV \\ gW \end{CD} & & & & \begin{CD} \rho_W \\ @VV{(g,\rho_W)}V \\ g\rho_W g^\dagger \end{CD} \\
& \rotatebox[origin=r]{-45}{$\overset{\varpi_{HS}}{\Verylongrightarrow}$} & \begin{CD} \rho_W \\ @VV{(g,\rho_W)}V \\ g\rho_W g^\dagger \end{CD} & \rotatebox{45}{$\Verylongeq$}
\end{array}
\end{equation}
$\mathscr Q^i$ is naturally equivalent to $\mathscr M^i \times \mathscr K^i$ by the trivialization functors $\phi_{\mathscr Q}^i \in \Funct(\mathscr M^i \times \mathscr K^i, \mathscr Q^i)$ and $\bar \phi_{\mathscr Q}^i \in \Funct(\mathscr Q^i,\mathscr M^i \times \mathscr K^i)$ defined by
\begin{equation}
\begin{array}{ccc}
\begin{CD} (\rho,k) \\ @V{(g,\rho,q,k)}VV \\ (g\rho g^\dagger,qk) \end{CD} & \overset{\phi_{\mathscr Q}^i}{\Verylongrightarrow} & \begin{CD} \sqrt \rho k \\ @VV{(\sqrt{g\rho g^\dagger}q \sqrt{\rho}^\pinv,\sqrt \rho k)}V \\ \sqrt{g\rho g^\dagger} qk \end{CD} \\
\\
\begin{CD} (WW^\dagger,\sqrt{WW^\dagger}^\pinv W) \\ @V{(g,WW^\dagger,\sqrt{gWW^\dagger g^\dagger}^\pinv g \sqrt{WW^\dagger},\sqrt{WW^\dagger}^\pinv W)}VV \\ (gWW^\dagger g^\dagger,\sqrt{gWW^\dagger g^\dagger}^\pinv gW) \end{CD} & \overset{\bar \phi_{\mathscr Q}^i}{\Verylongleftarrow} & \begin{CD} W \\ @VV{(g,W)}V \\ gW \end{CD}
\end{array}
\end{equation}
We can note that $\phi^i_{\mathscr Q}(\rho,k) = \phi^i_{\mathsf Q}(\rho,k)$. The arrows of the groupoid are $K^i \ltimes K^i$ because $(\sqrt{g\rho g^\dagger}^\pinv g \sqrt{\rho},k)$ models the transition from $K^i_{\sqrt \rho k}$ to $K^i_{g \sqrt \rho k}$ over the arrow $(g,\rho) \in \Morph \mathscr M^i$.\\

We remark that we can unify the categories: $\mathscr P = \bigsqcup_i \mathscr P^i$ and $\mathscr Q = \bigsqcup_i Q^i$, on which act the groupoids $\mathscr G = \bigsqcup_i \mathscr G^i$ and $\mathscr K =\bigsqcup \mathscr K^i$ (except that the arrows of different strata cannot be composed neither horizontally nor vertically). $\mathscr P$ and $\mathscr Q$ can be viewed as kinds of categorical principal bundles, but not strictly as principal 2-bundles\cite{Baez1,Baez2,Baez3,Wockel,Nikolaus,Viennot7} (the direct categorical generalization of the principal bundles) since in contrast with $\mathscr G^i$ and $\mathscr K^i$, $\mathscr G$ and $\mathscr K$ have not the structure of Lie crossed modules\cite{Baez1,Baez2,Baez3,Wockel,Nikolaus,Viennot7}.\\

Table \ref{categ} summarizes the roles of the different entities in the category formalism, and table \ref{categ2} summarizes the category structures.
\begin{table}
\caption{\label{categ} The different entities in the category formalism.}
\rotatebox{90}{\begin{tabular}{|r|c||c|c|l|}
\hline
\textit{entities} & \textit{symb.} & \textit{left role} & \textit{right role} & \textit{interpretation} \\
\hline \hline
left transformations & $g$ & object gauge changes & arrows of the categories & unitary/SLOCC transformations of $\mathcal S$ \\
\hline
right transformations & $k$ & arrows of the total category & object gauge changes & inobservable reconfigurations of $\mathcal A$ \\
\hline
stabilizers & $h$ & arrow gauge changes & inner arrows & inefficient transformations of $\mathcal S$ \\
\hline
diagonal matrices & $\sigma$ & objects of the base category & & statistical probabilities \\
\hline
density operators & $\rho$ & & objects of the base category & quantum mixed states \\
\hline
bounded operators & $W$ & objects of the total category & objects of the total category & quantum purified states \\
\hline
\end{tabular}}
\end{table}

\begin{table}
\caption{\label{categ2} The different category structures.}
\rotatebox{90}{\begin{tabular}{|r|c||c|c|l|}
\hline
\textit{Category} & \textit{symb.} & \textit{objects} & \textit{arrows} & \textit{interpretation} \\
\hline \hline
left base category & $\mathscr S^i$ & $\Sigma^i(\HS)$ & $\id_{\Sigma^i(\HS)}$ & space of the statistical probabilities \\
\hline
right base category & $\mathscr M^i$ & $\mathcal D^i(\HS)$ & $G \times \mathcal D^i(\HS)$ & quantum mixed state space with system transformations \\
\hline
left total category & $\mathscr P^i$ & $\mathcal B^i(\HS)$ & $\mathcal B^i(\HS) \times_{\mathcal B^i(\HS)} K^i$ & quantum purified state space with ancilla transformations \\
\hline
right total category & $\mathscr Q^i$ & $\mathcal B^i(\HS)$ & $G \times \mathcal B^i(\HS)$ & quantum purified state space with system transformations \\
\hline
left groupoid & $\mathscr G^i$ & $G$ & $G \rtimes \overline{H^i}$ & unitary and SLOCC transformations of the system \\
\hline
right groupoid & $\mathscr K^i$ & $K^i$ & $K^i \ltimes K^i$ &  inobservable reconfigurations of the ancilla\\
\hline
\end{tabular}}
\end{table}

\subsection{Connections}
\subsubsection{The connective structure of the left principal categorical bundle}
A connection on a principal 2-bundle with a base space being a trivial category is \cite{Baez1,Baez2,Baez3,Wockel} the data of a familly of connections on the objects related by a kind of gauge changes associated with the arrows.\\
More precisely, let $\mathcal B^i(\HS) = \Obj \mathscr P^i$ be viewed as a principal bundle over $\Sigma^i(\HS)$. $\forall W \in \mathcal B^i(\HS)$, the (left) vertical tangent space at $W$ is defined by
\begin{equation}
V_W^L \mathcal B^i(\HS) = \{XW, X \in \mathfrak g \setminus \mathfrak h^{i1} \}
\end{equation}
where $\mathfrak g$ is the Lie algebra associated with $G$ and $\mathfrak h^{i1}$ is the Lie algebra associated with $H^{i1}$ ($\mathfrak h^{i1} = \{0\}$ on the regular strata). We define the (left) horizontal tangent space at $W$ by
\begin{equation}
H^L_W \mathcal B^i(\HS) = \{\mathfrak X \in T_W \mathcal B^i(\HS), \mathfrak X W^\pinv = 0 \}
\end{equation}
On the regular strata, $H^L_W \mathcal B^i(\HS)$ is a connection for the $G$-principal bundle $\mathcal B^i(\HS)$. This connection is characterized by the connection 1-form $\omega_W^L \in \Omega^1(\mathcal B^i(\HS),\mathfrak g)$ defined by $\omega_W^L(\mathfrak X) = \mathfrak XW^{-1}$. On the singular strata, we consider the trivial $G$-bundle $G \times \Sigma^i(\HS)$ and we set an arbitrary connection $H^L_{(g,\sigma)} (G \times \Sigma^i(\HS))$ such that $\phi^i_{P*} \pi^{i1}_{G*} H^L_{(g,\sigma)} (G \times \Sigma^i(\HS)) = H^L_{\phi^i_{P}(\pi^{i1}_G(g),\sigma)} \mathcal B^i(\HS)$. We introduce $\omega^L_W \in \Omega^1(\mathcal B^i(\HS),\mathfrak g)$ such that $\pi^{i1*}_G \phi^{i*}_P \omega^L_W$ is the connection 1-form of $G \times \Sigma^i(\HS)$. Two different choices of arbitrary connections in $H^L_{(g,\sigma)} (G \times \Sigma^i(\HS))$ are related by $\tilde \omega^L_W - \omega^L_W = \eta_W \in \Omega^1(\mathcal B^i(\HS),\mathfrak h^{i1})$ which is a $\mathfrak h^{i1}$-gauge change of the second kind.\\
Let $\Phi : \Obj \mathcal P^i \to G$ be a left equivariant map from $\mathcal B^i(\HS)$ to $G$ ($\forall W \in \mathcal B^i(\HS)$, $\forall g \in G$, $\Phi(gW) = g\Phi(W)$). $\Phi$ defines a gauge transformation $a_\Phi : \Obj \mathscr P^i \to \Obj \mathscr P^i$ by $a_\Phi(W) = \Phi(W) W$. Under this gauge transformation, the connection 1-form becomes:
\begin{eqnarray}
a_\Phi^* \omega^L_W(\mathfrak X) & = & \omega^L_{\Phi(W)W} (a_{\Phi *} \mathfrak X) \\
& = & \omega^L_{\Phi(W)W} (\Phi(W)\mathfrak X) + \omega^L_{\Phi(W)W}(d\Phi(W) \Phi(W)^{-1} \Phi(W) W)
\end{eqnarray}
since $\forall \gamma(t)$ a curve on $\mathcal B^i(\HS)$ having $\mathfrak X$ as tangent vector at $t=0$ ($\gamma(0)=W$), we have $\frac{d}{dt} a_\Phi(\gamma(t)) = \frac{d}{dt} (\Phi(\gamma(t)) \gamma(t)) = \frac{d\Phi(\gamma(t))}{dt} \gamma(t) + \Phi(\gamma(t)) \frac{d\gamma(t)}{dt}$, and then $a_{\Phi *} \mathfrak X = \left. \frac{d}{dt} a_{\Phi}(\gamma(t)) \right|_{t=0} = d\Phi \Phi^{-1} \Phi W  + \Phi \mathfrak X$ ($d\Phi \Phi^{-1}: \mathcal B^i(\HS) \to \mathfrak g$). It follows that
\begin{equation}
a_\Phi^* \omega^L_W(\mathfrak X) = \Phi(W) \omega_W^L(\mathfrak X) \Phi(W)^{-1} + d\Phi(W)\Phi(W)^{-1}
\end{equation}
Let $\Upsilon : \Obj \mathscr P^i \to \bigcup_{W \in \mathcal B^i(\HS)} K^i_W$ be a right equivariant map ($\forall W \in \mathcal B^i(\HS)$, $\forall k \in K^i_W$, $\Upsilon(Wk) = \Upsilon(W)k$). $\Upsilon$ defines another kind of gauge transformation $a_\Upsilon : \Obj \mathcal P^i \to \Obj \mathcal P^i$ by $a_\Upsilon(W) = W\Upsilon(W)$. Under this gauge transformation of the second kind, the connection 1-form becomes
\begin{eqnarray}
a_\Upsilon^* \omega^L_W(\mathfrak X) & = & \omega^L_{W\Upsilon(W)} (a_{\Upsilon *} \mathfrak X) \\
& = & \omega^L_{W \Upsilon(W)}(\mathfrak X \Upsilon(W)) + \omega^L_{W \Upsilon(W)} (W \Upsilon(W) \Upsilon(W)^{-1} d\Upsilon(W))
\end{eqnarray}
since $\frac{d}{dt} a_\Upsilon(\gamma(t)) = \frac{d\gamma(t)}{dt} \Upsilon(\gamma(t)) + \gamma(t) \frac{d\Upsilon(\gamma(t))}{dt}$ and then $a_{\Upsilon *}\mathfrak X = \mathfrak X \Upsilon + W \Upsilon^{-1} d\Upsilon$ ($\Upsilon^{-1}d\Upsilon : \mathcal B^i(\HS) \to \mathfrak k^i_W$; $\mathfrak k^i_W$ being the Lie algebra of $K^i_W$). It follows that
\begin{equation}
a_\Upsilon^* \omega^L_W(\mathfrak X) = \omega^L_W(\mathfrak X) + W d\Upsilon(W) \Upsilon(W)^{-1} W^\pinv
\end{equation}
$(Wd\Upsilon \Upsilon^{-1} W^\pinv) WW^\dagger + WW^\dagger(Wd\Upsilon \Upsilon^{-1} W^\pinv)^\dagger = Wd\Upsilon \Upsilon^{-1} W + W \Upsilon d\Upsilon^{-1} W = 0$, we have then $Wd\Upsilon(W) \Upsilon(W)^{-1} W^\pinv \in \mathfrak{gl}(\HS)_{WW^\dagger} \subset \overline{\mathfrak h^i} = \bigoplus_{g\in G} g\mathfrak h^i g^{-1}$. The gauge transformation of the second kind is then a $\overline{\mathfrak h^i}$ gauge transformation. The connection on $\mathscr P^i$ is then not constituted by a single usual connection of $\mathcal B^i(\HS)$, but by a familly of usual connections of $\mathcal B^i(\HS)$ related by additions of $\overline{\mathfrak h^i}$-valued 1-forms.\\
We are now able to provide the local expressions of the connection. Let $s \in \Gamma(\Sigma^i(\HS),\mathcal B^i(\HS))$ be a local section ($s:\sigma \mapsto W_s(\sigma)$). We have
\begin{equation}
s^* \omega^L = dW_sW^\pinv_s = {^s}\mathfrak A \in \Omega^1(\Sigma^i(\HS),\mathfrak g)
\end{equation}
The $G$-gauge potential ${^s}\mathfrak A$ of the connection is the generator of the operator valued geometric phase (the connection having being defined for that purpose). Under a gauge change of the first kind $W_s(\sigma) \to g(\sigma) W_s(\sigma)$ (with $g \in \Omega^0(\Sigma^i(\HS),G)$), the gauge potential becomes:
\begin{equation}
{^{gs}}\mathfrak A = g{^s}\mathfrak Ag^{-1} + dg g^{-1}
\end{equation}
Under a gauge change of the second kind $W_s(\sigma) \to W_s(\sigma) k(\sigma)$ (with $k \in \Omega^i(\Sigma^i(\HS),K^i_{W(\sigma)})$), the gauge potential becomes:
\begin{equation}
{^{sk}} \mathfrak A = {^s}\mathfrak A + {^s}\eta_k
\end{equation}
with ${^s}\eta_k = W_s dk k^{-1} W_s^\pinv \in \Omega^1(\Sigma^i(\HS),\overline{\mathfrak h^i})$ (the $\overline{\mathfrak h^i}$-potential-transformation becomes under a gauge change of the first kind ${^{gs}}\eta_k = g {^s}\eta_k g^{-1}$). The arbitrary nature of the $\overline{\mathfrak h^i}$-potential-transformation is clear in the formalism (it is gauge change of the second kind).\\

The composite nature of the bundle $\mathcal B^i(\HS)$ permits also to consider an intermediate entity between the local and the global expression of the connection. Let $s^{\upcircarrow} \in \Gamma(\mathcal D^i(\HS), \mathcal B^i(\HS))$ be the local section ($s^{\upcircarrow}(\rho) = W_{s\upcircarrow}(\rho)$) such that $s^{\upcircarrow}(\rho) = \pi_{HS}(s(\pi_D(\rho)))$. The 1-form:
\begin{equation}
s^{\upcircarrow *} \omega^L = dW_{s\upcircarrow} W^\pinv_{s\upcircarrow} = {^{s \upcircarrow}}\mathfrak A \in \Omega^1(\mathcal D^i(\HS),\mathfrak g)
\end{equation}
satisfies for all $h \in \Omega^0(\mathcal D^i(\HS),\overline{H^i})$ such that $h(\rho) \in GL(\HS)_\rho = g_\rho H^i g_\rho^{-1}$ ($\phi^i_D(g_\rho,\pi_D(\rho)) = \rho$), the following gauge transformation rule:
\begin{equation}
{^{hs\upcircarrow}}\mathfrak A = h {^{s \upcircarrow}}\mathfrak A h^{-1}+ dh h^{-1}
\end{equation}

\begin{prop}
\label{generalgaugechange}
$\forall g \in \Omega^0(\mathcal D^i(\HS),G)$ we have
\begin{equation}
{^{s\upcircarrow}} \mathfrak A(g\rho g^\dagger) = g {^{s\upcircarrow}}\mathfrak A(\rho) g^{-1} + dgg^{-1} + g{^s}\eta(g,\rho) g^{-1}
\end{equation}
where ${^s}\eta \in \Omega^1(G \times \mathcal D^i(\HS),\overline{\mathfrak h^i})$ is a gauge change of the second kind. 
\end{prop}

\begin{proof}
$\forall \rho \in \mathcal D^i(\HS)$, $\exists k_\rho \in K$ such that $s^{\upcircarrow}(\rho) = W_{s\upcircarrow} = \sqrt{\rho} k_\rho$. In a first time, we consider the case of the global section $s^{\upcircarrow}(\rho) = \sqrt{\rho}$.
\begin{equation}
d\sqrt \rho \sqrt \rho + \sqrt \rho d\sqrt \rho = d\rho \Rightarrow d\sqrt \rho \sqrt \rho + (d\sqrt \rho \sqrt \rho)^\dagger = \frac{1}{2} (d\rho + d\rho^\dagger)
\end{equation}
\begin{eqnarray}
\Rightarrow \exists X_\rho, \text{ s.t. } X_\rho^\dagger = - X_\rho, & \quad & d\sqrt \rho \sqrt \rho = \frac{1}{2} d\rho + X_\rho \\
& & d\sqrt \rho \sqrt \rho^\pinv = \frac{1}{2} d\rho \rho^\pinv + X_\rho \rho^\pinv
\end{eqnarray}
We denote by ${^{\sqrt{.}\upcircarrow}}\mathfrak A(\rho) = d\sqrt{\rho} \sqrt{\rho}^\pinv$ the gauge potential for the global section, and by ${^{\id \upcircarrow}} \mathfrak A(\rho) = \frac{1}{2} d\rho \rho^\pinv$. By noting that $X_\rho \rho^\pinv \rho + \rho (X_\rho \rho^\pinv)^\dagger = 0$, we see that ${^{\sqrt{.}\upcircarrow}}\mathfrak A$ and ${^{\id \upcircarrow}} \mathfrak A$ are related by a $\overline{\mathfrak h^i}$-gauge change.
\begin{eqnarray}
{^{\sqrt{.}\upcircarrow}}\mathfrak A(g\rho g^\dagger) & = & {^{\id \upcircarrow}} \mathfrak A(g\rho g^\dagger) + X_{g\rho g^\dagger} g^{\dagger -1} \rho^\pinv g^{-1} \\
& = & \frac{1}{2} d(g\rho g^\dagger) (g\rho g^\dagger)^\pinv + X_{g\rho g^\dagger} g^{\dagger -1} \rho^\pinv g^{-1} \\
& = & \frac{1}{2} g d\rho \rho^\pinv g^{-1} + dg g^{-1} \nonumber \\
& & + g\left(\frac{1}{2} \rho dg^\dagger g^{\dagger -1} \rho^\pinv - \frac{1}{2} g^{-1} dg + g^{-1} X_{g\rho g^\dagger} g^{\dagger -1} \rho^\pinv \right) g^{-1} \\
& = & g {^{\sqrt{.}\upcircarrow}}\mathfrak A(\rho)g^{-1} + dg g^{-1} \nonumber \\
& & + g \underbrace{\left(\frac{1}{2} \rho dg^\dagger g^{\dagger -1} \rho^\pinv - \frac{1}{2} g^{-1}dg + g^{-1} X_{g\rho g^\dagger} g^{\dagger -1} \rho^\pinv - X_\rho \rho^\pinv \right)}_{\eta_1(g,\rho)} g^{-1}
\end{eqnarray}
$\eta_1 \rho + \rho \eta_1^\dagger = 0$, $\eta_1$ is a $\overline{\mathfrak h^i}$-gauge change.\\
We return to the general case:
\begin{equation}
{^{s\upcircarrow}} \mathfrak A(\rho) = dW_{s \upcircarrow} W_{s \upcircarrow}^\pinv = {^{\sqrt{.}\upcircarrow}}\mathfrak A(\rho) + \sqrt \rho dk_\rho k_\rho^{-1} \sqrt \rho^\pinv
\end{equation}
\begin{eqnarray}
{^{s\upcircarrow}} \mathfrak A(g \rho g^\dagger) & = & g {^{\sqrt{.}\upcircarrow}}\mathfrak A(\rho) g^{-1}+ dg g^{-1} + g\eta_1(g,\rho) g^{-1} + \sqrt{g \rho g^\dagger} dk_{g \rho g^\dagger} k^{-1}_{g \rho g^\dagger} \sqrt{g \rho g^\dagger}^\pinv \\
& = & g {^{s\upcircarrow}} \mathfrak A(\rho) g^{-1} + dgg^{-1} + g \eta_1(g,\rho) g^{-1} \nonumber \\
& & + \underbrace{\sqrt{g \rho g^\dagger} dk_{g \rho g^\dagger} k^{-1}_{g \rho g^\dagger} \sqrt{g \rho g^\dagger}^\pinv - g\sqrt{\rho} dk_{\rho} k^{-1}_{\rho} \sqrt{\rho}^\pinv g^{-1}}_{g \eta_2(g,\rho) g^{-1}}
\end{eqnarray}
\begin{eqnarray}
  \eta_2 \rho + \rho \eta_2^\dagger & =  & g^{-1} \sqrt{g\rho g^\dagger} dk_{g\rho g^\dagger} k_{g\rho g^\dagger} \sqrt{g\rho g^\dagger}^\pinv g \rho \nonumber \\
  & & - \sqrt{\rho} dk_\rho k_\rho^{-1} \sqrt{\rho} + \rho g^{\dagger} \sqrt{g \rho g^\dagger}^\pinv k_{g\rho g^\dagger} dk_{g\rho g^\dagger}^{-1} \sqrt{g\rho g^\dagger} g^{\dagger -1} \nonumber \\
  & & - \sqrt \rho k_\rho d k_\rho^{-1} \sqrt{\rho}
\end{eqnarray}
Since $kdk^{-1} = - dk k^{-1}$ and $\rho g^\dagger = g^{-1} \sqrt{g \rho g^\dagger} \sqrt{g \rho g^\dagger} \Rightarrow \rho g^\dagger \sqrt{g \rho g^\dagger}^\pinv = g^{-1} \sqrt{g \rho g^\dagger}$, we have $\eta_2 \rho + \rho \eta_2^\dagger = 0$, $\eta_2$ is a $\overline{\mathfrak h^i}$-gauge change.
\end{proof}
${^{s\upcircarrow}}\mathfrak A$ plays the roles of a $\overline{H^i}$-gauge potential and of a $G/H^i$-connection 1-form. The local and the intermediate expressions of the connection are related by
\begin{equation}
{^{s \upcircarrow}} \mathfrak A(\rho) = g_\rho {^s}\mathfrak A(\pi_D(\rho))g_\rho^{-1} + dg_\rho g_\rho^{-1} + g_\rho {^{\sqrt{.}}}\eta(g_\rho,\sigma) g_\rho^{-1}
\end{equation}
Conversely we have
\begin{equation}
{^s} \mathfrak A(\sigma) = (\pi_D^* {^{s \upcircarrow}} \mathfrak A)(\sigma)
\end{equation}

${^s}\mathfrak A$ defines the curving\cite{Baez1,Baez2,Baez3,Viennot1} of $\mathscr P^i$ as
\begin{equation}
{^s} B^L = d{^s}\mathfrak A - {^s}\mathfrak A \wedge {^s}\mathfrak A \in \Omega^2(\Sigma^i(\HS),\overline{\mathfrak h^i})
\end{equation}

\begin{prop}
${^s} B^L(\sigma) \in \overline{\mathfrak h^i}$
\end{prop}
\begin{proof}
Since $\mathfrak A WW^\dagger = dWW^\dagger$ we have $d\mathfrak A WW^\dagger - \mathfrak A \wedge d(WW^\dagger) = - dW \wedge dW^\dagger$. It follows that $d\mathfrak A WW^\dagger - \mathfrak A \wedge \mathfrak A WW^\dagger - \mathfrak A \wedge WW^\dagger \mathfrak A^\dagger = - dW \wedge dW^\dagger$. We have then $B^L WW^\dagger = \mathfrak A \wedge WW^\dagger \mathfrak A^\dagger + dW \wedge dW^\dagger$ and $WW^\dagger B^{L\dagger} = - \mathfrak A WW^\dagger \wedge \mathfrak A^\dagger + dW \wedge dW^\dagger$, implying that $B^L WW^\dagger + WW^\dagger B^{L\dagger} = 0 \Rightarrow B^L \in \mathfrak{gl}(\HS)_{WW^\dagger} \subset \overline{\mathfrak h^i}$.
\end{proof}

Remark, for a regular stratum, $\mathcal P^i$ is flat: ${^s} B^L = 0$ (${^s}\mathfrak A = dW_s W_s^{-1}$ is pure gauge). Under a gauge change of the first kind, the curving becomes
\begin{equation}
{^{gs}}B^L = g{^s}B^Lg^{-1}
\end{equation}
and under a gauge change of the second kind, it becomes
\begin{equation}
{^{sk}}B^L = {^s}B^L + d{^s}\eta_k - {^s}\eta_k \wedge {^s}\eta_k - [{^s}\mathfrak A,{^s}\eta_k]
\end{equation}

As in ref.\cite{Baez1,Baez2,Baez3,Viennot1} we can define fake curvatures of $\mathscr P^i$, as being ${^s}F^L = d{^s}A^L - {^s}A^L \wedge {^s}A^L - {^s}B^L \in \Omega^2(\Sigma^i(\HS),\mathfrak g_{\mathcal L})$ (${^s}A^L = \tr(X^{i\dagger} {^s}\mathfrak A) X_i$, $(X_i)_i$ being generators of $\mathfrak g_{\mathcal L}$) and ${^s}F^{BL} = d{^s}A^{BL} - {^s}A^{BL} \wedge {^s}A^{BL} - {^s}B^L \in \Omega^2(\Sigma^i(\HS),\overline{\mathfrak h^i})$ (${^s}A^{BL} = \mathcal P^\rho({^s}\mathfrak A)$, $\mathcal P^\rho$ being the projection of $\mathcal B(\HS)$ onto $\mathfrak{gl}(\HS)_\rho$). ${^s}F^L$ is the fake curvature of $\mathscr P^i$ with respect to the generator of the dynamics and ${^s}F^{BL}$ is the fake curvature of $\mathscr P^i$ with respect to the Sj\"oqvist-Andersson connection. For the adiabatic case, interpretations of $B^L$ and $F^L$ can be found in ref.\cite{Viennot3}, the curving is a measure of the ``kinematic decoherence'' (decoherence induced by variations in $\Sigma^i(\HS)$ during the dynamics), and the fake curvature is a measure of the non-adiabaticity of the system entangled with the ancilla. To understand these quantities in a non-adiabatic context, let $\Psi \in \Gamma(\Sigma^i(\HS),\HS \otimes \HA)$ be a local section of a singular stratum $i$. We denote by $p=(0,...,0,p_1,...,p_m) \in \Sigma(n)$ ($m < n = \dim \HS$, $p_i \in ]0,1[$) the local coordinates of $\sigma \in \Sigma^i(\HS)$. By a Schmidt decomposition \cite{Bengtsson} we can write that
\begin{equation}
\exists \phi^i(p) \in \HS, \exists \chi^i(p) \in \HA, \qquad \Psi(p) = \sum_{i=1}^m \sqrt{p_i} \phi^i(p) \otimes \chi^i(p)
\end{equation}
After some algebras, we can show that the curving associated with this section satisfies
\begin{equation}
\tr_{\HS}(\rho(p) B^L(p))  =  \sum_{i=1}^m p_i (F_{\mathcal A}(p) - A_{\mathcal S}(p) \wedge A_{\mathcal S}(p))_{ii}
\end{equation}
with $\rho(p) = \pi_{\mathcal A}(\Psi(p))$; $F_{\mathcal A}(p) = dA_{\mathcal A}(p) + A_{\mathcal A}(p) \wedge A_{\mathcal A}(p) \in \Omega^2(\Sigma^i(\HS),\mathfrak M_{m\times m}(\mathbb C))$ is the curvature associated with the Berry potential of the ancilla vectors of the Schmidt decomposition $A_{\mathcal A}(p)_{ij} = \langle \chi^i(p)|d\chi^j(p)\rangle_{\HA}$. Except non-abelian corrections, the statistical average of the curving $\omega_\rho(B^L)$ is essentially the average of the ancilla Berry curvature for the Schmidt decomposition. Let $\sigma^a,\sigma^b,\sigma^c \in \Sigma^i(\HS)$ be three infinitely close points. Let $\langle \sigma^a \sigma^b \sigma^c \rangle$ be the infinitesimal triangular simplex defined by these points. We have
\begin{equation}
e^{- \iint_{\langle \sigma^a \sigma^b \sigma^c \rangle} F_{\mathcal A}} \simeq Z^{bc} Z^{ca} Z^{ab}
\end{equation}
with $(Z^{ab})_{ij} = \langle \chi_i(p^a) | \chi_j(p^b) \rangle$. Now, if $\{\chi_i(p^b)\}_{i=1,...,m}$ spans the same subspace that $\{\chi_i(p^a)\}_{i=1,...,m}$, then $Z^{ab}$ is just a matrix of basis change inner to this subspace and $|\det(Z^{ab})|=1$. If the subspace spanned by $\{\chi_i(p^b)\}_{i=1,...,m}$ is different from the subspace spanned by $\{\chi_i(p^a)\}_{i=1,...,m}$, then $|\det(Z^{ab})|<1$. If follows that $\det(e^{- \iint_{\langle \sigma^a \sigma^b \sigma^c \rangle} F_{\mathcal A}})$ is a measure of the change of subspace in the neighbourhood of the simplex $\langle \sigma^a \sigma^b \sigma^c \rangle$. It follows that $F_{\mathcal A}(p)$ measures the propensity of the dynamics to leave the subspace spanned  by $\{\chi_i(p)\}_{i=1,...,m}$. Finally the curving $B^L(p)$ measures the propensity of the dynamics to leave the subspace spanned by $\{\phi_i(p) \otimes \chi_i(p)\}_{i=1,...,m}$. In order to describe the dynamics in the neighbourhood of $p$ (denoted by $\mathcal V(p)$), we need a subspace spanned by $\bigcup_{p' \in \mathcal V(p)} \{\phi_i(p') \otimes \chi_i(p')\}_{i=1,...,m}$. But even if $\mathcal V(p)$ is infinitesimal, if $\|B^L(p)\|$ is large, the dimension of this subspace will be larger than $m$. Finally we can understand $B^L(p)$ as the measure of the propensity of the system to leave its initial singular stratum to a less singular stratum. In accordance to this interpretation, $B^L$ is zero in regular strata. With similar arguments $F^L$ and $F^{BL}$ measure the propensity of the system to leave $S_{\mathcal L}(\Psi)$ and to leave a regime involving only phases by invariance.

\subsubsection{The connective structure of the right principal categorical bundle}
Let the following commutative diagramms:
$$ \begin{CD}
\Obj \mathscr Q^i @>{\id}>> \Morph \mathscr Q^i \\
@V{\varpi_{HS}}VV @VV{\varpi_{HS}}V \\
\Obj \mathscr M^i @>{\id}>> \Morph \mathscr M^i
\end{CD} 
\qquad
\begin{CD}
\Obj \mathscr Q^i @<{s}<< \Morph \mathscr Q^i \\
@V{\varpi_{HS}}VV @VV{\varpi_{HS}}V \\
\Obj \mathscr M^i @<{s}<< \Morph \mathscr M^i
\end{CD} 
\qquad
\begin{CD}
\Obj \mathscr Q^i @<{t}<< \Morph \mathscr Q^i \\
@V{\varpi_{HS}}VV @VV{\varpi_{HS}}V \\
\Obj \mathscr M^i @<{t}<< \Morph \mathscr M^i
\end{CD} 
$$
where we consider $\Obj \mathscr Q^i$ as a $K^i$-principal bundle on $\Obj \mathscr M^i$ and $\Morph \mathscr Q^i$ as $K^i \ltimes K^i$-principal bundle on $\Morph \mathscr M^i$. A connection on a principal 2-bundle with a base space being a non-trivial category is \cite{Viennot7} the data of connections $H_W \Obj \mathscr Q^i$ and $H_{(g,W)} \Morph \mathscr Q^i$ compatible in the sense where
\begin{eqnarray}
\id_{*} H_W \Obj \mathscr Q^i & = & H_{(1_\HA,W)} \Morph \mathscr Q^i \\
t_{*} H_{(g,W)} \Morph \mathscr Q^i & = & H_{gW} \Obj \mathscr Q^i \\
s_{*} H_{(g,W)} \Morph \mathscr Q^i & = & H_{W} \Obj \mathscr Q^i
\end{eqnarray}
the vertical tangent spaces being defined by $V_W \Obj \mathscr Q^i = V_W^R \mathcal B^i(\HS) = \{WX, X \in \mathfrak k^i_W\}$ and $V_{(g,W)} \Morph \mathscr Q^i =\{(0,WX), X \in \mathfrak k^i_W\}$. The connections are characterized by connection 1-forms, $\omega_o^R \in \Omega^1(\Obj \mathscr M^i,\mathfrak k^i)$ and $\omega_{\to}^R \in \Omega^1(\Morph \mathscr M^i, \mathfrak k^i \niplus \mathfrak k^i)$ related by
\begin{eqnarray}
\id^* \omega^R_\to & = & \omega_o^R \\
s^* \omega_o^R & = & \pi^{\mathfrak k}(\omega_o^R) \\
t^* \omega_o^R & = & t^{Lie}_\bullet (\omega^R_o)
\end{eqnarray}
with $\begin{array}{rcl} \pi^{\mathfrak k} :  \mathfrak k^i \niplus \mathfrak k^i & \to & \mathfrak k^i \\ Y \niplus X & \mapsto & X \end{array}$ and $\begin{array}{rcl} t^{Lie}_\bullet : \mathfrak k^i \niplus \mathfrak k^i & \to & \mathfrak k^i \\ Y \niplus X & \mapsto & Y+X \end{array}$. Let $\varsigma_o \in \Gamma(\Obj \mathscr M^i,\Obj \mathscr Q^i)$ and $\varsigma_\to \in \Gamma(\Morph \mathscr M^i, \Morph \mathscr Q^i)$ be the trivializing local sections:
\begin{eqnarray}
\varsigma_o(\rho) & = & \phi^i_{\mathscr Q}(\rho,1_\HS) = \sqrt{\rho} \\
\varsigma_{\to}(g,\rho) & = & \phi^i_{\mathscr Q}(g,\rho,1_\HA,1_\HA) = (\sqrt{g\rho g^\dagger} \sqrt{\rho}^\pinv,\sqrt \rho)
\end{eqnarray}
These local sections satisfy:
\begin{eqnarray}
id_{\varsigma_o(\rho)} & = & \varsigma_\to(\id_\rho) = (1_\HA,\sqrt \rho) \\
s(\varsigma_\to(g,\rho)) & = & \varsigma_o(s(g,\rho)) = \sqrt \rho \\
t(\varsigma_\to(h,\rho)) & = & \varsigma_o(t(g,\rho)) = \sqrt{g\rho g^\dagger}
\end{eqnarray}
The local data of the connection, the gauge potentials $A^R_o = \varsigma^*_o \omega^R_o \in \Omega^1(\Obj \mathscr M^i,\mathfrak k^i)$ and \\ $\underline A^R_\to = \varsigma^*_\to \omega^R_\to \in \Omega^1(\Morph \mathscr M^i, \mathfrak k^i \niplus \mathfrak k^i)$ are such that
\begin{eqnarray}
\id^* \underline A^R_\to(\rho) & = & A^R_o (\rho) \\
\pi^{\mathfrak k} (\underline A^R_\to (g,\rho)) & = & A^R_o(\rho) \\
t^{Lie}_\bullet (\underline A^R_\to (g,\rho)) & = & A^R_o(g\rho g^\dagger)
\end{eqnarray}
These conditions imply that $\underline A^R_\to(g,\rho) = A^R_o(\rho) + A^R_\to(g,\rho)$ with $A^R_\to(g,\rho) \in \Omega^1(\Morph \mathscr M^i, \mathfrak k^i)$ such that
\begin{eqnarray}
A^R_\to(1_\HS,\rho) & = & 0 \\
A^R_\to(g,\rho) & = & A^R_o(g\rho g^\dagger) - A^R_o(\rho)
\end{eqnarray}

In order to the horizontal lift of the arrows corresponding to the parallel transport of $\rho$ induces the geometric phases, we must have
$$
\begin{array}{ccccc}
\begin{CD} (\rho,1_\HA) \\ @V{(g_{\mathfrak A},\rho,1_\HA,1_\HA)}VV \\ (g_{\mathfrak A}\rho g^\dagger_{\mathfrak A},1_\HA) \end{CD} & \xrightarrow{\times (q,k)} \hspace{3em} & \begin{CD} (\rho,k) \\ @V{(g_{\mathfrak A},\rho, q,k)}VV \\ (g_{\mathfrak A}\rho g^\dagger_{\mathfrak A},qk) \end{CD} & \overset{\phi_{\mathscr Q}^i}{\Verylongrightarrow} & \begin{CD} g_{\mathfrak A}^{-1} \sqrt{g_{\mathfrak A} \rho g_{\mathfrak A}^\dagger} k^R \\ @VV{(g_{\mathfrak A},g_{\mathfrak A}^{-1} \sqrt{g_{\mathfrak A} \rho g_{\mathfrak A}^\dagger} k^R)}V \\ \sqrt{g_{\mathfrak A} \rho g_{\mathfrak A}^\dagger} k^R \end{CD}
\end{array}
$$
where $g_{\mathfrak A} = \Ted^{-\int_0^t i_{\dot{\sqrt{\rho}}} {^{Uhl \upcircarrow}}\mathfrak A dt'}$ is the left geometric phase, $k^R = \Teg^{-\int_0^t i_{\dot{\sqrt{\rho}}} A^R dt'}$ is the (right) Uhlmann geometric phase, and $(q,k) \in \Morph \mathscr K^i$ is the searched horizontal lift:
\begin{eqnarray}
(q,k) & = & \Teg^{-\int_0^t i_{\dot g_{\mathfrak A},\dot{\sqrt{\rho}}} \underline A^R_\to dt'} \\
& = & (\Teg^{-\int_0^t i_{\dot g_{\mathfrak A},\dot{\sqrt{\rho}}} A^R_\to dt'},\Teg^{-\int_0^t i_{\dot{\sqrt{\rho}}} \tilde A^R_o dt'})
\end{eqnarray}
where $i_{\dot{\sqrt{\rho}}} \tilde A^R_o = \left(\Teg^{-\int_0^t i_{\dot g_{\mathfrak A},\dot{\sqrt{\rho}}} A^R_\to dt'}\right)^{-1} i_{\dot{\sqrt{\rho}}} A^R_o \Teg^{-\int_0^t i_{\dot g_{\mathfrak A},\dot{\sqrt{\rho}}} A^R_\to dt'}$ by virtue of the intermediate representation theorem \cite{Messiah}. We have in the definition of the parallel transport $\sqrt{g_{\mathfrak A} \rho g_{\mathfrak A}^\dagger} k^R$ and not $g_{\mathfrak A} \sqrt \rho k^R$ as for the right connection, because we want that the target of a parallel transported arrow correspond to an usual parallel transport in the ``target bundle'', i.e. $\varsigma_o (g_{\mathfrak A} \rho g_{\mathfrak A}^\dagger) k^R = \sqrt{g_{\mathfrak A} \rho g_{\mathfrak A}^\dagger} k^R$.\\
By considering the definition of $\phi^i_{\mathscr Q}$ we find
\begin{equation}
\left(\sqrt{g_{\mathfrak A} \rho g_{\mathfrak A}^\dagger} \Teg^{-\int_0^t i_{\dot g_{\mathfrak A},\dot{\sqrt{\rho}}} A^R_\to dt'} \sqrt{\rho}^\pinv,\sqrt{\rho} \Teg^{-\int_0^t i_{\dot{\sqrt{\rho}}} \tilde A^R_o dt'} \right) = (g_{\mathfrak A},g_{\mathfrak A}^{-1} \sqrt{g_{\mathfrak A} \rho g_{\mathfrak A}^\dagger} k^R)
\end{equation}
inducing that
\begin{equation}
\begin{cases}
\Teg^{-\int_0^t i_{\dot g_{\mathfrak A},\dot{\sqrt{\rho}}} A^R_\to dt'} = \sqrt{g_{\mathfrak A} \rho g_{\mathfrak A}^\dagger}^\pinv g_{\mathfrak A} \sqrt \rho & \\
\Teg^{-\int_0^t i_{\dot g_{\mathfrak A},\dot{\sqrt{\rho}}} \underline A^R_\to dt'} = k^R
\end{cases}
\end{equation}
It follows from the last equality that $\underline A^R_\to(g,\rho) = A^R(g\rho g^\dagger) \iff A^R_\to(g,\rho) + A^R_o(\rho) = A^R(g\rho g^\dagger)$. We have clearly $A^R_o(\rho) = A^R(\rho)$, the ``object'' right connection is the Uhlmann connection. From the first equality, we see that the gauge potential of the ``arrow'' right connection is such that
\begin{eqnarray}
A^R_\to(g,\rho) & = & d\left(\sqrt{g \rho g^\dagger}^\pinv g \sqrt \rho \right) \left(\sqrt{g \rho g^\dagger}^\pinv g \sqrt \rho \right)^\pinv \\
& = & d\sqrt{g \rho g^\dagger}^\pinv \sqrt{g \rho g^\dagger} + \sqrt{g \rho g^\dagger}^\pinv dgg^{-1} \sqrt{g \rho g^\dagger} \nonumber \\
& & \quad + \sqrt{g \rho g^\dagger}^\pinv g {^{\sqrt{.}\upcircarrow}}\mathfrak A(\rho) g^{-1} \sqrt{g \rho g^\dagger}^\pinv \\
& = & \sqrt{g \rho g^\dagger}^\pinv \left(- {^{\sqrt{.}\upcircarrow}}\mathfrak A(g \rho g^\dagger) + dgg^{-1} g {^{\sqrt{.}\upcircarrow}}\mathfrak A(\rho) g^{-1} \right) \sqrt{g \rho g^\dagger}
\end{eqnarray}
${^{\sqrt{.}\upcircarrow}}\mathfrak A(g\rho g^\dagger) = g{^{\sqrt{.}\upcircarrow}}\mathfrak A g^{-1}+dgg^{-1} + g{^{\sqrt{.}}}\eta(g,\rho) g^{-1}$ by virtue of property \ref{generalgaugechange}. It follows that
\begin{equation}
A^R_\to(g,\rho) = - \sqrt{g \rho g^\dagger}^\pinv g {^{\sqrt{.}}}\eta(g,\rho) g^{-1} \sqrt{g \rho g^\dagger}
\end{equation}

The connections define the curving \cite{Viennot7} of $\mathscr Q^i$ is
\begin{equation}
B^R = dA^R_\to + A^R_\to \wedge A^R_\to + [A^R_o, A^R_\to] \in \Omega^2(\Morph \mathscr M^i, \mathfrak k^i)
\end{equation}
and the fake curvature of $\mathscr Q^i$ is
\begin{equation}
F^R = dA^R_o + A^R_o \wedge A^R_o + B^R \in \Omega^2(\Morph \mathscr M^i, \mathfrak k^i)
\end{equation}
$F^R(1_\HA,\rho)$ is just the curvature of the Uhlmann connection, it is a measure of the holonomy with respect to the Uhlmann's parallelism (relative phase factor \cite{Uhlmann0}) for infinitesimal loops starting from $\rho$. The Uhlmann's curvature is related to the quantum Fisher information matrix \cite{Matsumoto}, and it measures the difficulty to realize an estimation of the mixed state by measurements in the neighbourhood of $\rho$. The curving is just a correction associated with the shift between the two purifications of $g\rho g^\dagger$, i.e. $g \sqrt{\rho}$ and $\sqrt{g \rho g^\dagger}$.\\

We can note that it is possible to choose another ``object'' right connection as being another Uhlmann like connections. In particular, we can choose the Sj\"oqvist-Andersson connection by replacing $A^R$ by $A^{BR}$, in this case $k^R$ is the Sj\"oqvist geometric phase, and the rest of the discussion is totally similar. The fake curvature is then written $F^{BR} = dA^{BR} + A^{BR} \wedge A^{BR} + B^R$. The possibility of changing the ``object'' right connection follows from the definition of the parallel transport of the arrows $(g_{\mathfrak A},g_{\mathfrak A}^{-1} \sqrt{g_{\mathfrak A} \rho g_{\mathfrak A}^\dagger} k^R)$. If we change the definition of the right geometric phase $k^R$ (by passing, for example, from the Uhlmann to the Sj\"oqvist geometric phase), the arrows (in $\Morph \mathscr Q^i$ and $\Morph \mathscr M^i$) change in consequence, since $g_{\mathfrak A}$ is generated by $i_{\dot{\sqrt \rho}} \mathfrak A + \sqrt{\rho} \dot k^R k^{R-1} \sqrt{\rho}^\pinv$.

Table \ref{connection} summarizes the different data of the connective structures.
\begin{table}
\caption{\label{connection} The different data of the connective structures of $\mathscr P^i$ and $\mathscr Q^i$.}
\rotatebox{90}{\begin{tabular}{|r|c||c|c|c|c|l|}
\hline
\textit{differential form} & \textit{symb.} & \textit{side} &\textit{degree} & \textit{base space} & \textit{values} & \textit{interpretation} \\
\hline \hline
left gauge change of the 1st kind &  $g$ & L & 0 & $\Sigma^i(\HS)$ & $G$ & unitary and SLOCC operations on $\mathcal S$ \\
\hline
right gauge change &  $k$ & R & 0 & $\mathcal D^i(\HS)$ & $K$ & unobservable reconfigurations of $\mathcal A$ \\
\hline
left gauge change of the 2nd kind  &$\eta_k$ & L & 1 & $\Sigma^i(\HS)$ & $\overline{\mathfrak h^i}$ & unobservable reconfigurations of $\mathcal A$\\
\hline
left gauge potential &  $\mathfrak A$ & L& 1 & $\Sigma^i(\HS)$ & $\mathfrak g$ & generator of the operator valued geometric phase \\
\hline
potential-connection &  ${^{\upcircarrow}}\mathfrak A$ & L & 1 & $\mathcal D^i(\HS)$ & $\mathfrak g$ & generator of the operator valued geometric phase \\
\hline
right object gauge potential & $A^R_o$ & R & 1 & $\mathcal D^i(\HS)$ & $\mathfrak k^i$ & generator of the Uhlmann/Sj\"oqvist geometric phase \\
\hline
right arrow gauge potential &  $A^R_\to$ & R & 1 & $G \times \mathcal D^i(\HS)$ & $ \mathfrak k^i$ &  unitary and SLOCC operations on $\mathcal S$ \\
\hline
left curving & $B^L$ & L & 2 & $\Sigma^i(\HS)$ & $\overline{\mathfrak h^i}$ & measure of the kinematic decoherence \\
\hline
left fake curvature &  $F^L$ & L & 2 & $\Sigma^i(\HS)$ & $\mathfrak g_{\mathcal L}$ & measure of the transitions outsite $S_{\mathcal L}(\Psi)$  \\
\hline
two-sided left fake curvature  & $F^{BL}$ & L & 2 & $\Sigma^i(\HS)$ & $\overline{\mathfrak h^i}$ & measure of the non-invariance \\
\hline
right curving & $B^R$ & R & 2 & $G \times \mathcal D^i(\HS)$ & $\mathfrak k^i$ &  correction needed by the arrow structure\\
\hline
right fake curvature & $F^R$ & R &  2 & $G \times \mathcal D^i(\HS)$ & $ \mathfrak k^i$ & measure of the quantum estimation difficulty\\
\hline
two-sided right fake curvature & $F^{BR}$ & R &  2 & $G \times \mathcal D^i(\HS)$ & $ \mathfrak k^i$ & measure of the non-invariance \\
\hline
\end{tabular}}

\end{table}

\section{Conclusion}
The Lindblad equation describes a quantum system $\mathcal S$ in contact with a very large environment (a reservoir) $\mathcal R$. From a theoretical point of view, if $\Phi \in \HS \otimes \mathcal H_{\mathcal R}$ is the state of the system plus the reservoir, $\rho = \tr_{\mathcal H_{\mathcal R}} |\Phi \rrangle \llangle \Phi|$ obeys (under some assumptions\cite{Breuer}) to the Lindblad equation. But in practice, $\Phi$ is unknown because of the very large number of degrees of freedom of $\mathcal R$ (the partial trace models the forgetting of the ``informations'' concerning $\mathcal R$). The purification process permitts to introduce $\Psi \in \HS \otimes \HA$ such that $\rho = \tr_{\HA} |\Psi \rrangle \llangle \Psi|$, where the ancilla $\mathcal A$ plays the role of a small effective environment. Another approaches describing open quantum systems by a Schr\"odinger equation have been proposed \cite{Breuer, stocheq, Barchielli} (see also \ref{othereq}). These approaches involve stochastic processes in the Schr\"odinger equation in order to describe the effects of the environment onto the quantum system. These frameworks are not adaptated for our goal which consists to obtain a geometrization of the dynamics of open quantum systems, in order to use some methods issuing from the differential geometry. It is difficult to include random variables in a geometric framework, this induces the use of unsual geometries (based on stochastic and It\^o calculus). Our approach permitts to treat the dynamics of open quantum systems with the usual differential geometric methods. Moreover, with respect to the approaches consisting to model the environment with a high number of degrees of freedom, we propose a representation using an ancilly with minimal dimension as a kind of ``effective minimal environment''. This is important in practice for the use of the geometric tools on concrete example as for examples for the control of quantum systems hampered by the effects of the environment as in ref. \cite{Viennot3}. For the numerical computations of the geometric fields and for the size of the numerical data set to analyse, it is very important in practice to have a description involving the smallest Hilbert space.The price to pay to have a description in a small Hilbert space, without stochastic processes, is the nonlinearity of the Schr\"odinger equation governing $\Psi$. Moreover, in contrast with a precise description of the microscopic physics of the environment, the dynamics in our approach is still determined by the quantum jump operators and then by only the effects on the quantum system without knowledge of the environment. This is the reason of the gauge degree of freedom in the geometric description (the gauge changes interpreted as inobservable reconfigurations of the ancilla). Our approach can be then applied in situations where the physics of the environment is badly known and where the parameters of the Lindblad equation (jump rates, jump operators) are phenomenologically obtained from an experimental analysis.\\
The purified dynamics presented in this paper can permit to study the dynamics of open quantum systems submitted to relaxation processes, with the tools used for pure states. In particular we have shown how the geometric phases appear in this context, unifying some concepts previously introduced by different authors. The geometric structure involved by the purification and the geometric phases, is richer than the case of the pure states, and needs the use of the category theory. We hope that these facts can help to understand dynamics of open quantum systems and enlighten the position of the density matrix theory in the landscape of the geometric and gauge structures in theoretical physics.\\
In this paper, we have restricted our attention onto quantum systems described by finite dimensional Hilbert spaces. The extension of this work to infinite dimensional Hilbert spaces is an interesting question but it needs to treat some difficulties. Firstly, the various manifolds considered in this paper become infinite dimensional in that case. This involves the use of very special geometric methods with some topological cautions (for example the use of direct limits on chains of submanifolds with increasing dimension such that the union is the infinite dimensional manifold). Secondly, if the Hilbert space becomes infinite dimensional, the mixed states are not necessarily represented by a density operator. It is the case only if the relevant space of system observables is a von Neumann algebra. If it is just a $C^*$-algebra $\mathfrak a$, the mixed states are only defined as linear functionals of $\mathfrak a$ such that $\forall A \in \mathfrak a$, $\omega(A^\dagger A) \geq 0$ and $\|\omega\|=1$ ($\forall \omega$, $\exists \rho$ such that $\omega(A) = \tr(\rho A)$ if and only if the Hilbert space is finite dimensional or if $\mathfrak a$ is a von Neumann algebra) \cite{Bratteli}. In the two cases, the study of the Lindblad equation must be replaced by the study of transformation semigroups onto $C^*$ or von Neumann algebras. This needs a generalisation of the purification procedure in this context with anew some topological cautions.\\ 

The role of the category theory in quantum physics is an intriguing question, and it could be interesting to study the relations of the structure presented in this paper with the categorical structures proposed in other quantum problems\cite{Raptis, Coecke, Flori}. Moreover, it will be interesting in future works to generalize the holonomic quantum computation approach \cite{Zanardi, Lucarelli, Sjoqvist2, Xu} to open quantum systems by using the description based on the categorical bundles.\\

\noindent {\bf Acknowledgments :} The author acknowledges support from I-SITE Bourgogne-Franche-Comt\'e under grants from the I-QUINS project.

\appendix
\section{Comparison with other Schr\"odinger equations associated with the open quantum systems}
\label{othereq}
We have shown that the purified state obeys on the regular strata to the nonlinear Schr\"odinger equation:
\begin{equation}
\label{nlse}
\ihbar \dot \Psi_\rho = H^{eff} \otimes 1_\HA \Psi_\rho + \frac{\imath}{2} \gamma^k \Gamma_k \otimes \Gamma_k^\ddagger(\Psi_\rho) \Psi_\rho \qquad \Psi_\rho \in \HS \otimes \HA
\end{equation}
with $H^{eff} = H_{\mathcal S} - \frac{\imath}{2} \Gamma_k^\dagger \Gamma_k$. Ref. \cite{Gisin} studies a nonlinear Schr\"odinger equation for the pure evolution closest to the Lindblad evolution in the sense that it is viewed has the dynamics of some tangent vectors of $\mathcal D_0(\HS)$. This equation is the following:
\begin{equation}
\label{eqGisin}
\ihbar \dot \psi = (H^{eff} - \langle H^{eff} \rangle_\psi) \psi + \imath \gamma^k \langle \Gamma^\dagger_k \rangle_\psi (\Gamma_k - \langle \Gamma_k \rangle_\psi) \psi \qquad \psi \in \HS
\end{equation}
where $\langle O \rangle_\psi = \langle \psi|O|\psi \rangle$ ($\forall O \in \mathcal B(\HS)$). Except the shift of the operators by their average values, the structure of this equation is similar to eq. (\ref{nlse}). In fact, suppose that during a short time the dynamics remains on the singular stratum of the pure states: $\rho = |\psi \rangle \langle \psi|$, in that case $W_\rho = |\psi \rangle \langle \psi_0|$ and $W_\rho^\pinv = |\psi_0 \rangle \langle \psi|$ (with $\psi_0 = \psi(t=0)$). The equation (\ref{PNLSE}) becomes then
\begin{equation}
\label{nlsepure}
\ihbar \dot \psi \otimes \psi = (H^{eff} \psi + \frac{\imath}{2} \gamma^k \langle \Gamma_k^\dagger \rangle_\psi \Gamma_k \psi ) \otimes \psi
\end{equation}
which is very close to the equation (\ref{eqGisin}) (except the shift of the operators and a $\frac{1}{2}$ factor).\\
The deterministic part of the Liouville equation for a piecewise deterministic process is (see ref. \cite{Breuer} chapter 6.1) $\ihbar \dot \psi = H^{eff} \psi + \frac{\imath}{2} \gamma^k \|\Gamma_k \psi \|^2 \psi$. Since $\|\Gamma_k\|^2 = \langle \Gamma_k^\dagger \Gamma_k \rangle_\psi$, it is very close to our equation for a dynamics remaining on the pure state stratum. This is valid only between two quantum jumps; in a same manner equation (\ref{nlsepure}) could be approximatively valid only during a short time until the quantum jumps described by the operator $\frac{\imath}{2} \gamma^k \langle \Gamma_k^\dagger \rangle_\psi \Gamma_k$ tend to leave the pure state stratum. To describe quantum jumps we add a stochastic process (see ref. \cite{Breuer,stocheq}): $\psi \in \HS$
\begin{eqnarray}
\label{eqPDP}
\ihbar d\psi & = & H^{eff} \psi dt + \frac{\imath}{2} \gamma^k \langle \Gamma_k^\dagger \Gamma_k \rangle_\psi \psi dt + \imath dN^k_t \left(\frac{\Gamma_k}{\|\Gamma_k \psi\|} - 1_\HS\right) \psi \\
\text{or } \ihbar d\psi & = & H^{eff} \psi dt + \frac{\imath}{2} \gamma^k \left(\langle \Gamma_k^\dagger+\Gamma_k \rangle_\psi \Gamma_k - \frac{1}{4} \langle \Gamma_k^\dagger+\Gamma_k \rangle_\psi^2 \right) \psi dt \nonumber \\
& & \quad + \imath \sqrt{\gamma^k}dW_t^k \left(\Gamma_k -\frac{1}{2} \langle \Gamma_k^\dagger+\Gamma_k \rangle_\psi \right) \psi
\end{eqnarray}
where $dN^k_t$ is a Poisson process and $dW_t^k$ is a Wiener process. The first equation is the representation as a piecewise deterministic process and the second one is called quantum state diffusion process. The density matrix is then $\rho = \mathbb{E}\left(|\psi \rangle \langle \psi| \right)$ where $\mathbb E$ is the expectation value with respect to the stochastic process.\\
Ref. \cite{Yi} proposes a Schr\"odinger equation following the dynamics induced by the Lindblad equation:
\begin{equation}
\label{eqYi}
\ihbar \dot \Upsilon_\rho = H^{eff} \otimes 1_\HA \Upsilon_\rho - 1_\HS \otimes (H^{eff})^{\dagger \transp}\Upsilon_\rho + \imath \gamma^k \Gamma_k \otimes \Gamma_k^{\dagger \transp} \Upsilon_\rho \qquad \Upsilon_\rho \in \HS \otimes \HA
\end{equation}
(the operators of $\HS$ being transformed into operators of $\HA$ by the non canonical isomorphism ${O^a}_b |\zeta_a \rangle \langle \zeta^b| \mapsto {O^\alpha}_\beta |\xi_\alpha \rangle \langle \xi^\beta|$ ($(\zeta_a)_a$ and $(\xi_\alpha)_\alpha$ being the two choosen fixed basis of $\HS$ and $\HA$)). Equation (\ref{eqYi}) is linear, but it is not an equation for the purification, since in this case $\Upsilon_\rho = {\rho^a}_\beta \zeta_a \otimes \xi^\beta$ and then $\tr_{\HA} |\Upsilon_\rho \rrangle \llangle \Upsilon_\rho| \not= \rho$. In fact, if we choose the $\HA = \HS^{*}$, equation (\ref{eqYi}) is precisely the Lindblad equation in the Hilbert-Schmidt representation (in the Liouville space), $|\Upsilon_\rho\rrangle = {\rho^a}_b |\zeta_a\rangle \otimes \langle \zeta^b| = \rho$ (equation (\ref{eqYi}) is then just a reformulation of the ``superoperator'' formalism).\\
The comparisons between the different Schr\"odinger equations associated with an open quantum system are summarized table \ref{SEOQS}.
\begin{table}
\caption{\label{SEOQS} Properties of the Schr\"odinger equations associated with an open quantum system. ($^1$: an operator on the Hilbert-Schmidt space of operators is called superlinear in the physics literature.)}
\rotatebox{90}{\begin{tabular}{|l|c|l|c|l|}
\hline
\textit{Schr\"odinger equation} & \textit{Generator of the quantum jumps} & \textit{Nature} & \textit{dim.} & \textit{relation with the density matrix}\\
\hline
Hilbert-Schmidt representation & $\imath \gamma^k \Gamma_k \otimes \Gamma_k^{\dagger \transp}$ & ``superlinear''$^1$ & $n^2$ & $\rho = |\Upsilon_{\rho} \rrangle$ \\
\hline
Purification & $\frac{\imath}{2} \gamma^k \Gamma_k \otimes \Gamma_k^\ddagger(\Psi_\rho)$ & nonlinear & $n^2$ & $\rho = \tr_\HA |\Psi_\rho \rrangle\llangle \Psi_\rho|$ \\
\hline
Closest pure evolution & $\imath \gamma^k \langle \Gamma_k^\dagger \rangle_\psi \Gamma_k$ & nonlinear & $n$ & $\tr \left(\delta \rho - \delta|\psi\rangle \langle \psi|\right)^2$ is minimum\cite{Gisin} \\
\hline
Piecewise deterministic process & $dN^k_t \left(\frac{\Gamma_k}{\|\Gamma_k \psi\|} - 1_\HS\right)$ & nonlinear stochastic & $n$ & $\rho = \mathbb E(|\psi \rangle \langle \psi|)$ \\
\hline
Quantum state diffusion & $\imath \sqrt{\gamma^k}dW^k_t (\Gamma_k -\frac{1}{2} \langle \Gamma_k + \Gamma_k^\dagger \rangle_\psi)$ & nonlinear stochastic & $n$ & $\rho = \mathbb E(|\psi \rangle \langle \psi|)$ \\
\hline
\end{tabular}}
\end{table}

\end{document}